\newif\ifnotes
\newif\iffull
\newif\ifforus

\forustrue   
\fulltrue 
\notestrue

\iffull
\documentclass[11pt,letter]{article}
\usepackage{dtrt}
\usepackage{fullpage}
\else
\documentclass[conference]{IEEEtran}

\usepackage[\ifnotes notes=true \else notes=false \fi, later=false]{dtrt}

\fi

\usepackage{hyperref}
\usepackage{times}
\usepackage{amsmath}
\usepackage{amssymb}
\usepackage{amsthm}
\usepackage{comment}
\usepackage{xspace}
\usepackage{paralist}
\usepackage{enumitem}
  \setlist[itemize]{leftmargin=*}
  \setlist[enumerate]{leftmargin=*}
\usepackage{makecell}
\usepackage{tikz}
\usepackage{hypcap}
\usepackage{mleftright}
\usepackage{bbm}
\usepackage{colortbl}
\usepackage{tabu}
\usepackage{multirow}
\usepackage{mathrsfs}
\usepackage{adjustbox}
\usepackage{microtype}
\usepackage{etoolbox}
\usepackage{graphicx}
\usepackage{breakcites}
\usepackage{booktabs}
\usepackage{mdframed}
\usepackage{soul}
\usepackage[capitalise]{cleveref}
\usepackage{pifont}
\usepackage{multicol}
\usepackage[binary-units]{siunitx}
\usepackage{caption}
\usepackage{subcaption}
\usepackage{array}
\usepackage{seqsplit}
\usepackage{mathtools}
\usepackage{hhline}
\usepackage{dsfont}
\usepackage{textcomp}
\usepackage[ruled]{algorithm2e}
\usepackage{outlines}
\usepackage{mathrsfs}
\usepackage[mathscr]{euscript}
\usepackage{balance}
\usepackage{flushend}

\DeclareMathOperator*{\argmin}{arg\,min}

\definecolor{cerulean}{RGB}{109,155,195}

\newcommand{\sys}[0]{Helen\xspace}
\newcommand{\enc}[1]{\texttt{Enc}_{\text{PK}}(#1)}
\newcommand{\dec}[1]{\texttt{Dec}_{\text{SK}}(#1)}

\newcommand{\pk}{\text{PK}}
\newcommand{\sk}{\text{SK}}
\newcommand{\skshare}[1]{[\text{SK}]_{#1}}

\renewcommand\vec{\mathbf}
\newcommand{\matr}[1]{\mathbf{#1}}
\newcommand{\adv}[0]{\mathscr{A}}

\definecolor{ashgrey}{rgb}{0.7, 0.75, 0.71}

\newcommand{\newgadget}[1]{
\begin{mdframed}[backgroundcolor=ashgrey!30]
\begin{gadget}
#1
 \end{gadget}
 \end{mdframed}
}

\newtheorem{theorem}{Theorem}
\newtheorem{gadget}{Gadget}
\newtheorem{mydef}{Definition}
\newtheorem{lemma}{Lemma}
\newtheorem{protocol}{Protocol}
\newtheorem{innercustomthm}{Theorem}
\newenvironment{customthm}[1]
  {\renewcommand\theinnercustomthm{#1}\innercustomthm}
  {\endinnercustomthm}

\makeatletter
\newcommand{\removelatexerror}{\let\@latex@error\@gobble}
\makeatother

\newcommand{\eat}[1]{}

\newcommand{\justforus}[1]{
\ifforus
\textcolor{cerulean}{================ Just for us ==============}

#1

\textcolor{cerulean}{=======================================}
\fi
}

\begin{document}

\title{   { \fontsize{19}{20} \selectfont \sys: Maliciously Secure Coopetitive Learning for Linear Models}             }

\author{Wenting Zheng, Raluca Ada Popa, Joseph E. Gonzalez, and Ion Stoica \\ UC Berkeley}

\maketitle


\begin{abstract}
Many organizations wish to collaboratively train machine learning models on their combined datasets for a common benefit (e.g., better medical research, or fraud detection). 
However, they often cannot share their plaintext datasets due to privacy concerns and/or business competition.
In this paper, we design and build \sys, a system that allows multiple parties to train a linear model without revealing their data, a setting we call  \emph{coopetitive learning}. 
Compared to prior secure training systems, \sys protects against a much stronger adversary who is {\em malicious} and can compromise $m-1$ out of $m$ parties.
Our evaluation shows that \sys can achieve up to five orders of magnitude of performance improvement when compared to training using an existing state-of-the-art secure multi-party computation framework. 
\end{abstract}

\section{Introduction}\label{sec:intro}


Today, many organizations are interested in training machine learning models over their aggregate sensitive data.
The parties also agree to release the model to every participant so that everyone can benefit from the training process.
In many existing applications, collaboration is advantageous because training on more data tends to yield higher quality models~\cite{Halevy09}.
Even more exciting is the potential of enabling new applications that are not possible to compute  using a single party's data because they require training on complementary data from multiple parties (e.g., geographically diverse data).
However, the challenge is that these organizations cannot share their sensitive data in plaintext due to privacy policies and regulations~\cite{hipaa} or due to business competition~\cite{stoica2017berkeley}. 
We denote this setting using the term {\em coopetitive learning}\footnote{We note that Google uses the term {\em federated learning}~\cite{stoica2017berkeley} for a different but related setting: a semi-trusted cloud trains a model over the data of millions of user devices, which are intermittently online, and sees sensitive intermediate data.  }, where the word ``coopetition''~\cite{coopetition} is  a portmanteau of ``cooperative'' and ``competitive''.
To illustrate coopetitive learning's potential impact as well as its challenges, we summarize two concrete use cases.

\smallskip\noindent\textit{A banking use case.}
The first use case was shared with us by two large banks in North America. 
Many banks want to use machine learning to detect money laundering more effectively.
Since criminals often hide their traces by moving assets across different financial institutions, an accurate model would require training on data from different banks.
Even though such a model would  benefit all participating banks, these banks cannot share their customers' data in plaintext because of privacy regulations and business competition.

\smallskip\noindent\textit{A medical use case.}
The second use case was shared with us by a major healthcare provider who needs to distribute vaccines during the annual flu cycle.
In order to launch an effective vaccination campaign (i.e., sending vans to vaccinate people in remote areas), this organization would like to identify areas that have high probabilities of flu outbreaks using machine learning. 
More specifically, this organization wants to train a linear model over data from seven geographically diverse medical organizations.
Unfortunately, such training is impossible at this moment because the seven organizations cannot share their patient data with each other due to privacy regulations.
\begin{figure}[t!]
    \centering
    \includegraphics[width=0.4\textwidth]{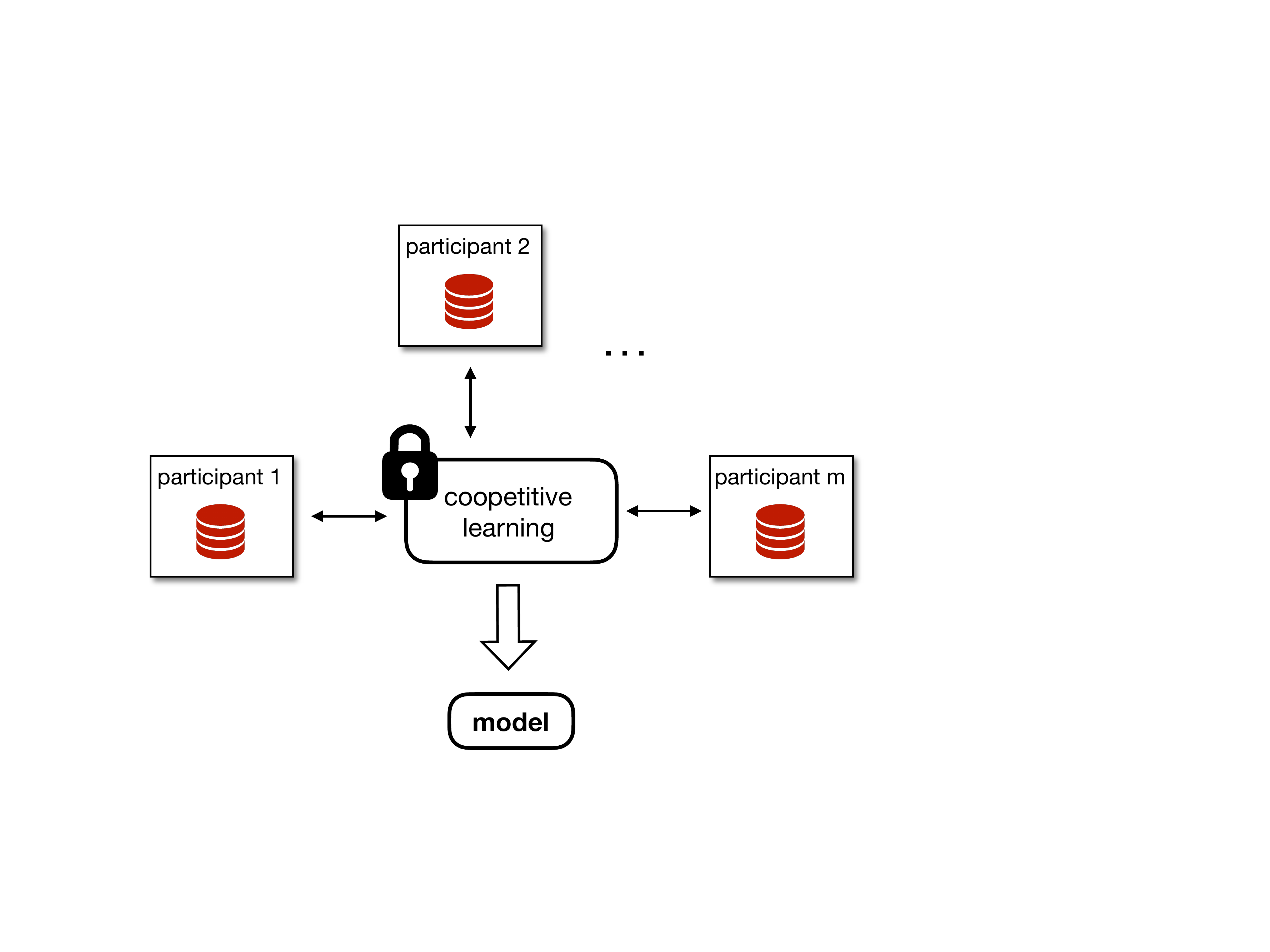}
    \caption{The setting of coopetitive learning. 
    \vspace{-3mm}}
    \label{fig:intro}
\end{figure}

The general setup of coopetitive learning fits within the cryptographic framework of secure multi-party computation (MPC)~\cite{ben1988completeness,Goldreich:1987,yao1982protocols}.
Unfortunately, implementing training using generic MPC frameworks is extremely inefficient, so recent training systems~\cite{Nikolaenko,HallRidge,secureml,DistributedGascon,Cock2015,cryptoeprint:2017:979,quadraticopt} opt for tailored protocols instead.
However, many of these systems rely on outsourcing to  non-colluding servers, and all assume a passive attacker who never deviates from the protocol.
In practice, these assumptions are often not realistic because they essentially require  an organization to base the confidentiality of its data on the correct behavior of other organizations.
In fact, the banks from the aforementioned use case informed us that they are not comfortable with trusting the behavior of their competitors when it comes to sensitive business data. 

Hence, we need a much stronger security guarantee:  each organization  should {\em only  trust  itself.} 
This goal calls for  maliciously secure MPC in the setting where $m-1$ out of $m$ parties can fully misbehave. 

In this paper, we design and build~\sys, a platform for maliciously secure coopetitive learning.
\sys supports a significant slice of machine learning and statistics problems: regularized linear models. 
This family of models includes  ordinary least squares regression, ridge regression, and LASSO. 
Because these models are statistically robust and easily interpretable, they are widely used in cancer research~\cite{lassoForCancer}, genomics~\cite{lassoForGenomics1, lassoForGenomics2}, financial risk analysis~\cite{lassoForFinance, lassoForCreditScoring}, and are the foundation of basis pursuit techniques in signals processing.

The setup we envision for \sys is similar to the use cases above:  a few  organizations (usually less than $10$) have large amounts of data (on the order of  hundreds of thousands or millions of records) with a smaller number of features (on the order of tens or hundreds).

While it is possible to build such a system by implementing a standard training algorithm like Stochastic Gradient Descent (SGD)~\cite{sgd} using a generic maliciously secure MPC protocol, the result is very inefficient. 
To evaluate the practical performance difference, we implemented SGD using SPDZ, a maliciously secure MPC library~\cite{SPDZgithub}. 
For a configuration of $4$ parties, and a real dataset of $100$K data points per party and $90$ features, such a baseline can take an estimated time of 3 months to train a linear regression model.
Using a series of techniques explained in the next section, \sys can train the same model in less than 3 hours.

\subsection{Overview of techniques}

To solve such a challenging problem, \sys combines insights from cryptography, systems, and machine learning.
This synergy enables an efficient and scalable solution under a strong threat model.
One recurring theme in our techniques is that, while the overall training process needs to scale linearly with the total number of training samples, the more expensive \emph{cryptographic} computation can be reformulated to be \emph{independent} of the number of samples.

\newcommand{\mypara}[1]{\smallskip\noindent\textbf{#1}}

Our first insight is to leverage a classic but under-utilized technique in distributed convex optimization called Alternating Direction Method of Multipliers (ADMM)~\cite{ADMM}.
The standard algorithm for training models today is SGD, which optimizes an objective function by iterating over the input dataset.
With SGD, the number of iterations scales at least linearly with the number of data samples. 
Therefore, na\"ively implementing SGD using a generic MPC framework would require an expensive MPC synchronization protocol for every iteration.
%
Even though ADMM is less popular for training on plaintext data, we show that it is much more efficient for cryptographic training  than SGD. 
One advantage of ADMM is that it converges in very few iterations (e.g., a few tens) because each party repeatedly solves local optimization problems.
Therefore, utilizing ADMM allows us to dramatically reduce the number of MPC synchronization operations.
Moreover, ADMM is very efficient in the context of linear models because the local optimization problems can be solved by closed form solutions.
These  solutions are also easily expressible in cryptographic computation and are especially efficient because they operate on small summaries of the input data that only scale with the dimension of the dataset.

However, merely expressing ADMM in MPC does not solve an inherent scalability problem.
As mentioned before, \sys addresses a strong threat model in which an attacker can deviate from the protocol.
This malicious setting requires the protocol to ensure that the users' behavior is correct.
To do so, the parties need to commit to their input datasets and prove that they are consistently using the same datasets throughout the computation.
A na\"ive way of solving this problem is to have each party commit to the entire input dataset and calculate the summaries using MPC.
This is problematic because 1) the cryptographic computation will scale linearly in the number of samples, and 2) calculating the summaries would also require \sys to calculate complex matrix inversions within MPC (similar to~\cite{nikolaenko2013privacy}).
Instead, we make a second observation that each party can use singular value decomposition (SVD)~\cite{golub2012matrix} to decompose its input summaries into small matrices that scale only in the number of features.
Each party commits to these decomposed matrices and proves their properties using matrix multiplication to avoid explicit matrix inversions.

Finally, one important aspect of ADMM is that it enables decentralized computation.
Each optimization iteration consists of two phases: \emph{local optimization} and \emph{coordination}.
The local optimization phase requires each party to solve a local sub-problem.
The coordination phase requires all parties to synchronize their local results into a single set of global weights.
Expressing both phases in MPC would encode local optimization into a computation that is done by every party, thus losing the decentralization aspect of the original protocol.
Instead, we observe that the local operations are all linear matrix operations between the committed summaries and the global weights.
Each party knows the encrypted global weights, as well as its own committed summaries in plaintext.
Therefore, \sys uses partially homomorphic encryption to encrypt the global weights so that each party can solve the local problems in a decentralized manner, and enables each party to efficiently prove in zero-knowledge that it computed the local optimization problem correctly. 

\section{Background}

\subsection{Preliminaries}
In this section, we describe the notation we use for the rest of the paper.
Let $P_1, ..., P_m$ denote the $m$ parties.
Let $\mathds{Z}_N$ denote the set of integers modulo $N$, and $\mathds{Z}_p$ denote the set of integers modulo a prime $p$.
Similarly, we use $\mathds{Z}^{*}_{N}$ to denote the multiplicative group modulo $N$.

We use $z$ to denote a scalar, $\vec{z}$ to denote a vector, and $\matr{Z}$ to denote a matrix.
We use $\enc{x}$ to denote an encryption of $x$ under a public key $\text{PK}$.
Similarly, $\dec{y}$ denotes a decryption of $y$ under the secret key $\text{SK}$.

Each party $P_i$ has a feature matrix $\matr{X}_i \in \mathds{R}^{n \times d}$, where $n$ is the number of samples per party and $d$ is the feature dimension.  $\vec{y}_i \in \mathds{R}^{n \times 1}$
 is the labels vector.
The machine learning datasets use floating point representation, while our cryptographic primitives use groups and fields.
Therefore, we represent the dataset using fixed point integer representation.

\subsection{Cryptographic building blocks}
In this section, we provide a brief overview of the cryptographic primitives used in \sys.

\subsubsection{Threshold partially homomorphic encryption}
A partially homomorphic encryption scheme is a public key encryption scheme that allows limited computation over the ciphertexts.
For example, Paillier~\cite{Paillier} is an additive homomorphic encryption scheme: multiplying two ciphertexts together (in a certain group) generates a new ciphertext such that its decryption yields the sum of the two original plaintexts.
Anyone with the public key can encrypt and manipulate the ciphertexts based on their homomorphic property.
This  encryption scheme also acts as a perfectly binding and computationally hiding homomorphic commitment scheme~\cite{groth2009homomorphic}, another  property we  use in \sys.

A \emph{threshold} variant of such a scheme has some additional properties.
While the public key is known to everyone, the secret key is split across a set of parties such that a subset of them must participate together to decrypt a ciphertext.
If not enough members participate, the ciphertext cannot be decrypted.
The threshold structure can be altered based on the adversarial  assumption.
In \sys, we use a threshold structure where \emph{all} parties must participate in order to decrypt a ciphertext.

\subsubsection{Zero knowledge proofs}
Informally, zero knowledge proofs are proofs that prove that a certain statement is true without revealing the prover's secret for this statement. 
For example, a prover can prove that there is a solution to a Sudoku puzzle without revealing the actual solution.
Zero knowledge \emph{proofs of knowledge} additionally prove that the prover indeed knows the secret.
\sys uses modified $\Sigma$-protocols~\cite{damgaard2002sigma} to prove properties of a party's local computation.
The main building blocks we use are ciphertext proof of plaintext knowledge, plaintext-ciphertext multiplication, and ciphertext interval proof of plaintext knowledge~\cite{cramer2001multiparty,boudot2000efficient}, as we further explain in \cref{s:gadgets}.
Note that $\Sigma$-protocols are honest verifier zero knowledge, but can be transformed into full zero-knowledge using existing techniques~\cite{damgaard2000efficient,faust2012non,garay2003strengthening}.
In our paper, we present our protocol using the $\Sigma$-protocol notation.

\subsubsection{Malicious MPC} 
We utilize SPDZ~\cite{spdz}, a state-of-the-art malicious MPC protocol, for both \sys and the secure baseline we evaluate against.
Another recent malicious MPC protocol is authenticated garbled circuits~\cite{wang2017global}, which supports boolean circuits.
We decided to use SPDZ for our baseline because the majority of the computation in SGD is spent doing matrix operations, which is not efficiently represented in boolean circuits.
For the rest of this section we give an overview of the properties of SPDZ.


An input $a \in \mathds{F}_{p^{k}}$ to SPDZ is represented as $\langle a \rangle = (\delta, (a_1, \dots, a_n), (\gamma(a)_1, \dots, \gamma(a)_n))$, where $a_i$ is a share of $a$ and $\gamma(a)_i$ is the MAC share authenticating $a$ under a SPDZ global key $\alpha$.
Player $i$ holds $a_i, \gamma(a)_i$, and $\delta$ is public. 
During a correct SPDZ execution, the following property must hold: $a = \sum_i a_i$ and $\alpha(a + \delta) = \sum_i \gamma(a)_i$.
The global key $\alpha$ is not revealed until the end of the protocol; otherwise the malicious parties can use $\alpha$ to construct new MACs.

SPDZ has two phases: an offline phase and an online phase.
The offline phase is independent of the function and generates precomputed values that can be used during the online phase, while the online phase executes the designated function.


\subsection{Learning and Convex Optimization}
\label{sec:lasso}


    





Much of contemporary machine learning can be framed in the context of minimizing the \emph{cumulative error} (or loss) of a model over the training data. 
While there is considerable excitement around deep neural networks, the vast majority of real-world machine learning applications still rely on robust linear models because they are well understood and can be efficiently and reliably learned using established convex optimization procedures.

In this work, we focus on linear models with squared error and various forms of regularization resulting in the following set of multi-party optimization problems:
\begin{equation}
\hat{\vec{w}} = \arg\min_w \frac{1}{2}\sum_{i=1}^m  \| \matr{X}_i \vec{w} - \vec{y}_i\|^2_2 + \lambda \mathbf{R}(\vec{w}),\label{eqn:leastsquares}
\end{equation}
where $\matr{X}_i \in \mathbb{R}^{n \times d}$ and $\vec{y}_i \in \mathbb{R}^{n}$ are 
the training data (features and labels)
from party $i$.
The regularization function $\mathbf{R}$ and regularization tuning parameter $\lambda$ are used to improve prediction accuracy on high-dimensional data.  
Typically, the regularization function takes one of the following forms:
\begin{align}
    \mathbf{R}_{L^1}(\vec{w}) & = \sum_{j=1}^d |\vec{w}_j|, &
    \mathbf{R}_{L^2}(\vec{w}) & = \frac{1}{2}\sum_{j=1}^d \vec{w}_j^2 \nonumber
\end{align}
corresponding to Lasso ($L^1$) and Ridge ($L^2$) regression respectively.
The estimated model $\hat{\vec{w}} \in \mathbb{R}^d$ can  then be used to render a new prediction $\hat{y}_* = \hat{\vec{w}}^T \vec{x}_*$ at a query point $\vec{x}_*$.
It is worth noting that in some applications of LASSO (e.g., genomics~\cite{lassoForGenomics1}) the dimension $d$ can be larger than $n$.  
However, in this work we focus on settings where $d$ is smaller than $n$, and the real datasets and scenarios we use in our evaluation satisfy this property.


\parhead{ADMM} Alternating Direction Method of Multipliers (ADMM)~\cite{ADMM} is an established technique for distributed convex optimization.  
%
To use ADMM, we first reformulate Eq.~\ref{eqn:leastsquares} by introducing \emph{additional} variables and constraints:
\begin{align}
 \underset{
 \{\vec{w}_i\}_{i=1}^m, \, \vec{z}
 }{\textbf{minimize:}}
  & \quad \frac{1}{2}\sum_{i=1}^m
  \|\matr{X}_i \vec{w}_i - \vec{y}_i\|_2^2 + \lambda \mathbf{R}(\vec{z}),\nonumber \\
\textbf{such that: } &\quad \vec{w}_i = \vec{z} \text{ for all $i \in \{1,\ldots, p\}$} \label{eqn:constrainedobj}
\end{align}

This equivalent formulation splits $\vec{w}$ into $\vec{w}_i$ for each party $i$, but still requires that $\vec{w}_i$ be equal to a global model $\vec{z}$.
To solve this constrained formulation, we construct an \emph{augmented Lagrangian}:
\begin{align}
    & L\left(\{\vec{w}_i\}_{i=1}^m, \vec{z}, \vec{u}\right) = 
     \frac{1}{2}\sum_{i=1}^m  \|\matr{X}_i \vec{w}_i - \vec{y}_i\|_2^2 + \lambda \mathbf{R}(\vec{z}) + \nonumber \\
    & \quad\quad \rho\sum_{i=1}^m \vec{u}_i^T \left(\vec{w}_i  - \vec{z}\right) +
    \frac{\rho}{2}\sum_{i=1}^m \left|\left|\vec{w}_i  - \vec{z}\right|\right|_2^2, \label{eqn:lagrangian}
\end{align}
where the dual variables $\vec{u}_i \in \mathbb{R}^d$ capture the mismatch between the model estimated by party $i$ and the global model $\vec{z}$ and the augmenting term $\frac{\rho}{2}\sum_{i=1}^m \left|\left|\vec{w}_i  - \vec{z}\right|\right|_2^2$ adds an additional penalty (scaled by the constant $\rho$) for deviating from $\vec{z}$.

The ADMM algorithm is a simple iterative dual ascent on the augmented Lagrangian of~\cref{eqn:constrainedobj}.
On the $k^{\text th}$ iteration, each party locally solves this closed-form expression:
%
\begin{equation}
    \vec{w}_i^{k+1} \leftarrow \left(\matr{X}_i^T \matr{X}_i + \rho \matr{I}\right)^{-1} 
    \left( \matr{X}_i^T \vec{y}_i + \rho \left(\vec{z}^k  - \vec{u}_i^k\right)\right) \label{eq:primalupdate}
\end{equation}
and then shares its local model $\vec{w}_i^{k+1}$ and Lagrange multipliers $\vec{u}_i^{k}$ to solve for the new global weights:
\begin{align}
    \vec{z}^{k+1} & \leftarrow \arg \min_{\vec{z}}  \lambda \mathbf{R}(\vec{z}) +  \frac{\rho}{2}\sum_{i=1}^m ||\vec{w}_i^{k+1} - \vec{z} + \vec{u}_i^{k} ||_2^2. \label{eq:zupdate}
\end{align}
Finally, each party uses the new global weights $\vec{z}^{k+1}$ to update its local Lagrange multipliers
\begin{align}
    \vec{u}_i^{k+1} & \leftarrow \vec{u}_i^{k} + \vec{w}_i^{k+1} - \vec{z}^{k+1} .\label{eq:dualupdate}
\end{align}

%
The update equations \eqref{eq:primalupdate}, \eqref{eq:zupdate}, and \eqref{eq:dualupdate} are executed iteratively until all updates reach a fixed point.  
In practice, a fixed number of iterations may be used as a stopping condition, and that is what we do in \sys.

\parhead{LASSO}
\newcommand{\shared}[1]{\textcolor{blue}{#1}}
We use LASSO as a running example for the rest of the paper in order to illustrate how our secure training protocol works.
LASSO is a popular regularized linear regression model that uses the $L^1$ norm as the regularization function.
The LASSO formulation is given by the optimization objective $\argmin_{\vec{w}} \| \matr{X} \vec{w} - \vec{y}\|_{2}^{2} + \lambda \| \vec{w} \|$.
The boxed section below shows the ADMM training procedure for LASSO.
Here, the quantities in \shared{color} are quantities that are intermediate values in the computation and need to be protected from every party, whereas the quantities in black are private values known to one party. 


\begin{framed}
\label{fig:lasso}
\underline{The coopetitive learning task for LASSO} 

\medskip

\noindent Input of party $P_i$: $\matr{X}_i, \vec{y}_i$

\begin{enumerate}
\item $\matr{A}_i \gets  \left(\matr{X}_i^T \matr{X}_i + \rho \matr{I}\right)^{-1} $

\item $\vec{b}_i \gets \matr{X}_i^T \vec{y}_i $

\item $\vec{u^0}, \vec{z^0}, \vec{w^0} \gets \vec{0}$

\item {\tt For} $k = 0$, {\tt ADMMIterations}-1:

\begin{enumerate}
\item  $\shared{\vec{w}_i^{k+1}} \leftarrow  \matr{A}_i (\vec{b}_i + \rho \left(\shared{\vec{z}^k}  - \shared{\vec{u}_i^k}\right))$
  
\item   $\shared{\vec{z}^{k+1}} \leftarrow 
    S_{\lambda/m\rho}
    \left(\frac{1}{m}\sum_{i=1}^m \left(\shared{\vec{w}_i^{k+1}} + \shared{\vec{u}_i^{k}} \right) \right)$

\item    $\shared{\vec{u}_i^{k+1}}  \leftarrow \shared{\vec{u}_i^{k}} + \shared{\vec{w}_i^{k+1}} - \shared{\vec{z}^{k+1}}$ 
\end{enumerate}
\end{enumerate}

\end{framed}

$S_{\lambda/m\rho}$ is the soft the soft thresholding operator, where
\begin{equation}
S_{\kappa}(a)=
\begin{cases}
a - \kappa & a > \kappa\\
0 & |a| \leq \kappa \\
a + \kappa & a < -\kappa
\end{cases} \label{eq:softthreshold}
\end{equation}

The parameters $\lambda$ and $\rho$ are public and fixed.

\section{System overview}

\begin{figure*}[t!]
    \centering
    \includegraphics[width=0.8\textwidth]{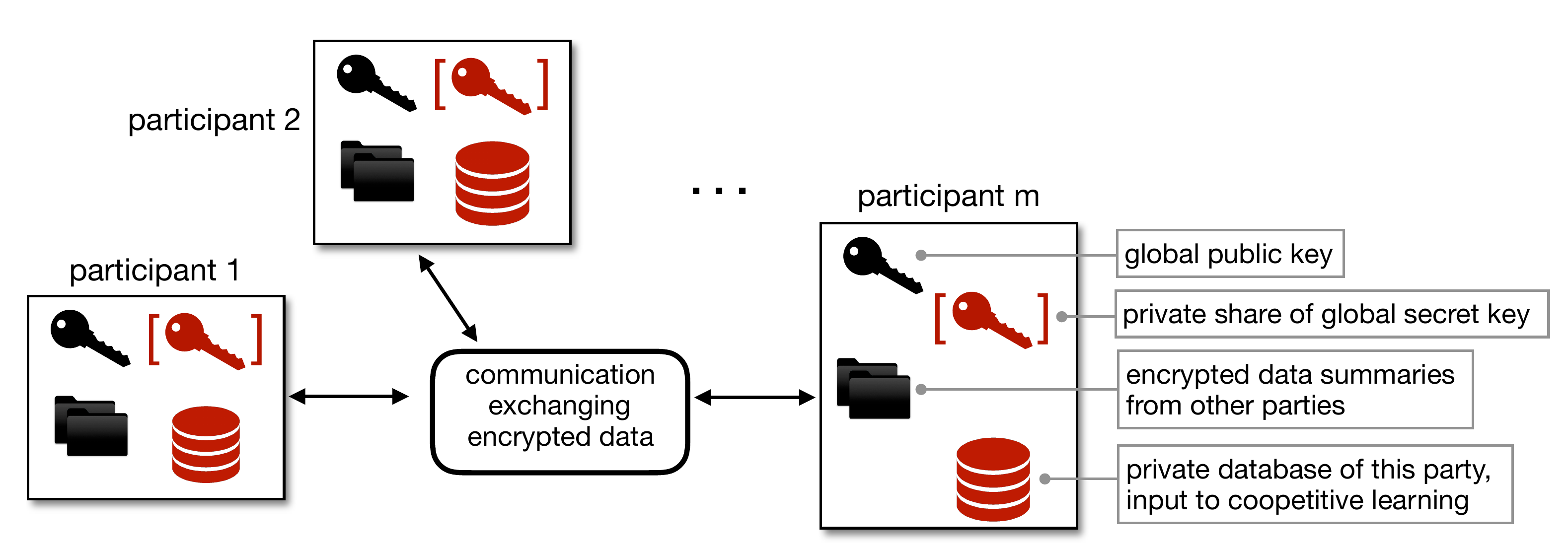}
    \caption{Architecture overview of \sys.  Every red shape indicates secret information known only to the indicated party, and black indicates public information visible to everyone (which could be private information in encrypted form). For participant $m$, we annotate the meaning of each quantity. }
    \label{fig:overview}
\end{figure*}

Figure~\ref{fig:overview} shows the system setup in \sys. 
A group of $m$ participants (also called parties) wants to jointly train a model on their data without sharing the plaintext data.
As mentioned in \cref{sec:intro}, the use cases we envision for our system consist of a few large organizations (around $10$ organizations), where each organization has a lot of data ($n$ is on the order of hundreds of thousands or millions). The number of features/columns in the dataset $d$ is on the order of tens or hundreds. Hence $d \ll n$. 

We assume that the parties have agreed to publicly release the final model.
As part of \sys, they will engage in an interactive protocol during which they share encrypted data, and only at the end will they obtain the model in decrypted form. 
\sys supports regularized linear models including least squares linear regression, ridge regression, LASSO, and elastic net.
In the rest of the paper, we focus on explaining \sys via LASSO, but we also provide update equations for ridge regression in~\cref{sec:applications}.

\subsection{Threat model}\label{s:threat}

We assume that all parties have agreed upon a single functionality to compute and have also consented to releasing the final result of the function to every party.

We consider a strong threat model in which all but one party can be compromised by a malicious attacker.
This means that the compromised parties can deviate arbitrarily from the protocol, such as supplying inconsistent inputs, substituting their input with another party's input, or executing different computation than expected. 
In the flu prediction example, six divisions could collude together to learn information about one of the medical divisions. 
However, as long as the victim medical division follows our protocol correctly, the other divisions will not be able to learn anything about the victim division other than the final result of the function.
We now state the security theorem.

\begin{customthm}{6}
\sys securely evaluates an ideal functionality 
$f_{\text{ADMM}}$ in the $(f_{\text{crs}}, f_{\text{SPDZ}})$-hybrid model under standard cryptographic assumptions, against a malicious adversary who can statically corrupt up to $m-1$ out of $m$ parties.
\end{customthm}

We formalize the security of \sys in the standalone MPC model. 
$f_{\text{crs}}$ and $f_{\text{SPDZ}}$ are ideal functionalities that we use in our proofs, where $f_{\text{crs}}$ is the ideal functionality representing the creation of a common reference string, and $f_{\text{SPDZ}}$ is the ideal functionality that makes a call to SPDZ. 
\iffull
We present the formal definitions as well as proofs in~\cref{sec:proofs}.
\fi

\smallskip\noindent\textbf{Out of scope attacks/complementary directions.}
\sys does not prevent a malicious party from choosing a bad dataset for the coopetitive computation (e.g., in an attempt to alter the computation result). 
In particular, \sys does not prevent poisoning attacks~\cite{jagielski2018manipulating,chen2017targeted}.
MPC protocols generally do not protect against bad inputs because there is no way to ensure that a party provides true data.
Nevertheless, \sys will ensure that once a party supplies its input into the computation, the party is bound to using the same input consistently throughout the entire computation; in particular, this prevents a party from providing different inputs at different stages of the computation, or mix-and-matching inputs from other parties.
Further,  some additional constraints can also be placed in pre-processing, training, and post-processing to mitigate such attacks, as we 
 elaborate  in~\cref{sec:ml:attacks}.

\sys also does not protect against attacks launched on the public model, for example, attacks that attempt to recover the training data from the model itself~\cite{shmatikovmachine,carlini2018secret}.
The parties are responsible for deciding if they are willing to share with each other the model. 
Our goal is only to conduct this computation securely: to ensure that the parties do not share their raw plaintext datasets with each other, that they do not learn more information than the resulting model, and that only the specified computation is executed.
Investigating techniques for ensuring that the model does not leak too much about the data is a complementary direction to \sys, and we expect that many of these techniques could be plugged into a system like \sys. 
 For example, 
 \sys can be easily combined with some differential privacy tools that add noise before model release to ensure that the model does not leak too much about an individual record in the training data.
We further discuss  possible approaches  in~\cref{sec:diffpriv}.


Finally, \sys does not protect against denial of service -- all parties must participate in order to produce a model.

\subsection{Protocol phases} \label{sec:phases}

We now explain the protocol phases at a high level. 
The first phase requires all parties to agree to perform the coopetitive computation, which happens before initializing \sys.
The other phases are run using \sys. 

\parhead{Agreement phase.}
In this phase, the $m$ parties come together and agree that they are willing to run a certain learning algorithm (in \sys's case, ADMM for linear models) over their joint data. 
The parties should also agree to release the computed model among themselves.

The following discrete phases are run by \sys.
We summarize their purposes here and provide the technical design for each in the following sections. 

\parhead{Initialization phase.}
During initialization, the $m$ parties compute the threshold encryption parameters~\cite{fouque2000sharing} using a generic maliciously secure MPC protocol like SPDZ~\cite{spdz}.
The public output of this protocol is a public key $\pk$ that is known to everyone.
Each party also receives a piece (called a \emph{share}) of the corresponding secret key $\sk$: party $P_i$ receives the $i$-th share of the key denoted as $\skshare{i}$. 
A value encrypted under $\pk$ can only be decrypted via all shares of the $\sk$, so every party needs to agree to decrypt this value. 
Fig.~\ref{fig:overview} shows these keys.
This phase only needs to run once for the entire training process, and does not need to be re-run as long as the parties' configuration does not change.

\parhead{Input preparation phase.} 
In this phase, each party prepares its data for the coopetitive computation. 
Each party $P_i$ precomputes summaries of its data and commits to them by broadcasting encrypted summaries to all other parties. 
The parties also need to prove that they know the values inside these encryptions using zero-knowledge proofs of knowledge. 
From this moment on, party $P_i$ will not be able to use different inputs for the rest of the computation.

By default, each party stores the encrypted summaries from other parties.
This is a viable solution since these summaries are much smaller than the data itself.
It is possible to also store all $m$ summaries in a public cloud by having each party produce an integrity MAC of the summary from each other party and checking the MAC upon retrieval which protects against a compromised cloud.

\parhead{Model compute phase.} This phase follows the iterative ADMM algorithm, in which parties successively compute locally on encrypted data, followed by a coordination step with other parties using a generic MPC protocol.

Throughout this protocol, each party receives only encrypted intermediate data. 
No party learns the intermediate data because, by definition, an MPC protocol should not reveal any data beyond the final result.
Moreover, each party proves in zero knowledge to the other parties that it performed the local computation correctly using data that is consistent with the private data that was committed in the input preparation phase. 
If any one party misbehaves, the other parties will be able to detect the cheating with overwhelming probability.

\parhead{Model release phase.} At the end of the model compute phase, all parties obtain an encrypted model.
All parties jointly decrypt the weights and release the final model.
However, it is possible for a set of parties to not receive the final model at the end of training if other parties misbehave (it has been proven that it is impossible to achieve fairness for generic MPC in the malicious majority setting~\cite{mpcfairness}).
Nevertheless, this kind of malicious behavior is easily detectable in \sys and can be enforced using legal methods.


\section{Cryptographic Gadgets}\label{s:gadgets}

\sys's design combines several different cryptographic primitives.
In order to explain the design clearly, we split \sys into modular gadgets.
In this section and the following sections, we discuss (1) how \sys implements these gadgets, and (2) how \sys composes them in the overall protocol.

For simplicity, we present our zero knowledge proofs as $\Sigma$-protocols, which require the verifier to generate random challenges.
These protocols can be transformed into full zero knowledge with non-malleable guarantees with existing techniques~\cite{garay2003strengthening,faust2012non}.
\iffull
We explain one such transformation in~\cref{sec:proofs}.
\fi


\subsection{Plaintext-ciphertext matrix multiplication proof}
\label{prot:paillier:matmul:proof}


\newgadget{\label{gadget:knowsone}
A zero-knowledge proof for the statement: 
 ``Given public parameters:
    public key $PK$, encryptions $E_\vec{X}$, $E_\vec{Y}$ and $E_\vec{Z}$;
    private parameters: $\vec{X}$,
 \begin{itemize}
     \item $\dec{E_\vec{Z}} = \dec{E_\vec{X}} \cdot \dec{E_\vec{Y}}$, and
     \item I know $\vec{X}$ such that $\dec{E_\vec{X}} = \vec{X}$.''
 \end{itemize}
}


\parhead{Gadget usage} 
We first explain how \cref{gadget:knowsone} is used in \sys.
A party $P_i$ in \sys knows a plaintext $\matr{X}$ and commits to $\matr{X}$ by publishing its encryption, denoted by $\enc{\matr{X}}$.
$P_i$ also receives an encrypted matrix $\enc{\matr{Y}}$ and needs to compute $\enc{\matr{Z}} = \enc{\matr{XY}}$ by leveraging the homomorphic properties of the encryption scheme. 
Since parties in \sys may be malicious, other parties cannot trust $P_i$ to compute and output $\enc{\matr{Z}}$ correctly. 
\cref{gadget:knowsone} will help $P_i$ prove in zero-knowledge that it executed the computation correctly.
The proof needs to be zero-knowledge so that nothing is leaked about the value of $\matr{X}$. 
It also needs to be a proof of knowledge so that $P_i$ proves that it knows the plaintext matrix $\matr{X}$.

\parhead{Protocol}
Using the Paillier ciphertext multiplication proofs~\cite{cramer2001multiparty}, we can construct a na\"ive algorithm for proving matrix multiplication.
For input matrices that are $\mathds{R}^{l \times l}$, the na\"ive algorithm will incur a cost of $l^{3}$ since one has to prove each individual product.
One way to reduce this cost is to have the prover prove that  $\vec{t} \matr{Z} = (\vec{t} \matr{X}) \matr{Y}$ for a randomly chosen $\vec{t}$ such that $\vec{t}_i = t^{i} \mod q $ (where $\vec{t}$ is a challenge from the verifier).
For such a randomly chosen $t$, the chance that the prover can construct a $\vec{t}\matr{Z}^{'} = \vec{t}\matr{X}\matr{Y}$ is exponentially small 
\iffull
(see~\cref{proof:paillier:matmul:proof} for an analysis).
\fi

As the first step, both the prover and the verifier apply the reduction to get the new statement $\enc{\vec{t} \matr{Z}} = \enc{\vec{t} \matr{X}} \enc{\vec{Y}}$.
To prove this reduced form, we apply the Paillier ciphertext multiplication proof in a straightforward way.
This proof takes as input three ciphertexts: $E_{a}, E_{b}, E_{c}$.
The prover proves that it knows the plaintext $a^{*}$ such that $a^{*} = \dec{E_a}$, and that $\dec{E_c} = \dec{E_a} \cdot \dec{E_b}$.
We apply this proof to every multiplication for each dot product in $(\vec{t} \matr{X}) \cdot \matr{Y}$. 
The prover then releases the individual encrypted products along with the corresponding ciphertext multiplication proofs. 
The verifier needs to verify that $\enc{\vec{t} \matr{Z}} = \enc{\vec{t} \matr{X} \matr{Y}}$.
Since the encrypted ciphers from the previous step are encrypted using Paillier, the verifier can homomorphically add them appropriately to get the encrypted vector $\enc{\vec{t} \matr{X} \matr{Y}}$.
From a dot product perspective, this step will sum up the individual products computed in the previous step.
Finally, the prover needs to prove that each element of $\vec{t} \matr{Z}$ is equal to each element of $\vec{t} \matr{X} \matr{Y}$.
We can prove this using the same ciphertext multiplication proof by setting $a^{*} = 1$.

\subsection{Plaintext-plaintext matrix multiplication proof}

\newgadget{\label{gadget:knowall}
A zero-knowledge proof for the statement: 
 ``Given public parameters: public key $PK$, encryptions $E_\vec{X}$, $E_\vec{Y}$, $E_\matr{Z}$; private parameters: $\vec{X}$ and $\vec{Y}$,
 \begin{itemize}
     \item $\dec{E_\vec{Z}} = \dec{E_\vec{X}} \cdot \dec{E_\vec{Y}}$, and
     \item I know $\vec{X}$, $\vec{Y}$, and $\vec{Z}$ such that $\dec{E_\vec{X}} = \vec{X}$, $\dec{E_\vec{Y}} = \vec{Y}$, and $\dec{E_\vec{Z}} = \vec{Z}$.''
 \end{itemize}
}

\parhead{Gadget usage} 

This proof is used to prove matrix multiplication when the prover knows \emph{both} input matrices (and thus the output matrix as well).
The protocol is similar to the plaintext-ciphertext proofs, except that we have to do an additional proof of knowledge of $\matr{Y}$.


\parhead{Protocol}
The prover wishes to prove to a verifier that $\matr{Z} = \matr{X} \matr{Y}$ without revealing $\matr{\matr{X}}, \matr{\matr{Y}}$, or $\matr{\matr{Z}}$. 
We follow the same protocol as~\cref{gadget:knowsone}.
Additionally, we utilize a variant of the ciphertext multiplication proof that only contains the proof of knowledge component to show that the prover also knows $\matr{Y}$.
The proof of knowledge for the matrix is simply a list of element-wise proofs for $\matr{Y}$.
We do not explicitly prove the knowledge of $\matr{Z}$ because the matrix multiplication proof and the proof of knowledge for $\matr{Y}$ imply that the prover knows $\matr{Z}$ as well.


\section{Input preparation phase}
\label{sec:offline:phase}

\subsection{Overview}
In this phase, each party prepares data for coopetitive training. 
In the beginning of the ADMM procedure, every party precomputes some summaries of its data and commits to them by broadcasting encrypted summaries to all the other parties. 
These summaries are then reused throughout the model compute phase.
Some form of commitment is necessary in the malicious setting because an adversary can deviate from the protocol by altering its inputs.
Therefore, we need a new gadget that allows us to efficiently commit to these summaries.

\newcommand{\commit}{\mathsf{Comm}}
\newcommand{\commitx}{\mathsf{Comm}^X}
\newcommand{\commity}{\mathsf{Comm}^y}

More specifically, the ADMM computation reuses two matrices during training: $\vec{A}_i = (\vec{X}_i^{T} \vec{X}_i + \rho \vec{I})^{-1}$ and $\vec{b}_i =  \vec{X}_i^{T} \vec{y}_i$ from party $i$ (see~\cref{sec:lasso} for more details).
These two matrices are of sizes $d \times d$ and $d \times 1$, respectively.
In a semihonest setting, we would trust parties to compute $\matr{A}_i$ and $\vec{b}_i$ correctly.
In a malicious setting, however, the parties can deviate from the protocol and choose $\matr{A}_i$ and $\vec{b}_i$ that are inconsistent with each other (e.g., they do not conform to the above formulations).

\sys does not have any control over what data each party contributes because the parties must be free to choose their own $\vec{X}_i$ and $\vec{y}_i$.
However, \sys ensures  that each party consistently uses the same $\matr{X}_i$ and $\vec{y}_i$ during the entire protocol. 
Otherwise, malicious parties could try to use different/inconsistent $\vec{X}_i$ and $\vec{y}_i$ at different stages of the protocol, and thus manipulate the final outcome of the computation to contain the data of another party. 

One possibility to address this problem is for each party $i$ to commit to its $\vec{X}_i$ in $\enc{\matr{X}_i}$ and $\vec{y}_i$ in $\enc{\vec{y}_i}$. 
To calculate $\matr{A}_i$, the party can calculate and prove $\matr{X}_i^{T} \matr{X}$ using~\cref{gadget:knowall}, followed by computing a matrix inversion computation within SPDZ.
The result $\matr{A}_i$ can be repeatedly used in the iterations.
This is clearly inefficient because (1) the protocol scales linearly in $n$, which could be very large, and (2) the matrix inversion computation requires heavy compute.

Our idea is to prove using an alternate formulation via {\em singular value decomposition (SVD)}~\cite{golub2012matrix}, which can be much more succinct: $\matr{A}_i$ and $\vec{b}_i$ can be decomposed using SVD to matrices that scale linearly in $d$.
Proving the properties of $\matr{A}_i$ and $\vec{b}_i$ using the decomposed matrices is equivalent to proving using $\matr{X}_i$ and $\vec{y}_i$.

\subsection{Protocol}
\subsubsection{Decomposition of reused matrices}
We first derive an alternate formulation for $\matr{X}_i$ (denoted as $\matr{X}$ for the rest of this section).
From fundamental linear algebra concepts we know that every matrix has a corresponding singular value decomposition~\cite{golub2012matrix}.
More specifically, there exists unitary matrices $\matr{U}$ and $\matr{V}$, and a diagonal matrix $\matr{\Gamma}$ such that
$\matr{X} = \matr{U} \matr{\Gamma} \matr{V}^T$, where $\matr{U} \in \mathds{R}^{n \times n}$, $\matr{\Gamma} \in \mathds{R}^{n \times d}$, and $\matr{V} \in \mathds{R}^{d \times d}$. 
Since $\matr{X}$ and thus $\matr{U}$ are real matrices, the decomposition also guarantees that $\matr{U}$ and $\matr{V}$ are orthogonal, meaning that $\matr{U}^T \matr{U} = \matr{I}$ and $\matr{V}^T \matr{V} = \matr{I}$.
%
If $\matr{X}$ is not a square matrix, then the top part of $\matr{\Gamma}$ is a diagonal matrix, which we will call $\matr{\Sigma} \in \mathds{R}^{d \times d}$.
$\matr{\Sigma}$'s diagonal is  a list of singular values $\sigma_i$.
The rest of the $\matr{\Gamma}$ matrix are $0$'s.
If $\matr{X}$ is a square matrix, then $\matr{\Gamma}$ is simply $\matr{\Sigma}$.
Finally, the matrices $\matr{U}$ and $\matr{V}$ are orthogonal matrices.
Given an orthogonal matrix $\matr{Q}$, we have that $\matr{Q} \matr{Q}^T = \matr{Q}^T \matr{Q} = \matr{I}$.

It turns out that $\matr{X}^{T} \matr{X}$ has some interesting properties:
\begin{eqnarray*}
\matr{X}^{T} \matr{X} &=& (\matr{U} \matr{\Gamma} \matr{V}^T)^T \matr{U} \matr{\Gamma} \matr{V}^T \\ 
&=& \matr{V} \matr{\Gamma}^{T} \matr{U}^{T} \matr{U} \matr{\Gamma} \matr{V}^{T} \\ 
&=& \matr{V} \matr{\Gamma}^T  \matr{\Gamma} \matr{V}^T  \\
&=& \matr{V} \matr{\Sigma}^2 \matr{V}^T. 
\end{eqnarray*}
We now show that $(\matr{X}^T \matr{X} + \rho \matr{I})^{-1} = \matr{V} \matr{\Theta} \matr{V}^{T}$, where $\matr{\Theta}$ is the diagonal matrix with diagonal values  $\dfrac{1}{\sigma_i^{2} + \rho}$.

\begin{eqnarray*}
(\matr{X}^T \matr{X} + \rho \matr{I}) \matr{V} \matr{\Theta} \matr{V}^T &=&  
\matr{V} (\matr{\Sigma}^2 + \rho \matr{I}) \matr{V}^T  \matr{V} \matr{\Theta} \matr{V}^T \\ 
&=& \matr{V} (\matr{\Sigma}^2 +\rho \matr{I}) \matr{\Theta} \matr{V}^{T}\\
&=& \matr{V} \matr{V}^{T} = \matr{I}.
\end{eqnarray*}



Using a similar reasoning, we can also derive that 
\begin{eqnarray*}
\matr{X}^{T} \matr{y} = \matr{V} \matr{\Gamma}^T \matr{U}^T \vec{y}.
\end{eqnarray*}

\subsubsection{Properties after decomposition}
The SVD decomposition formulation sets up an alternative way to commit to matrices $(\matr{X}_i^{T} \matr{X}_i + \rho \matr{I})^{-1}$ and $\matr{X}_i \vec{y}_i$.
For the rest of this section, we describe the zero knowledge proofs that every party has to execute.
For simplicity, we focus on one party and use $\matr{X}$ and $\vec{y}$ to represent its data, and $\matr{A}$ and $\vec{b}$ to represent its summaries.

During the ADMM computation, matrices $\matr{A} = (\matr{X}^T \matr{X} + \rho \matr{I})^{-1}$ and $\matr{b} = \matr{X}^T \vec{y}$ are repeatedly used to calculate the intermediate weights. 
Therefore, each party needs to commit to $\matr{A}$ and $\matr{b}$.
With the alternative formulation, it is no longer necessary to commit to $\matr{X}$ and $\vec{y}$ individually.
Instead, it suffices to prove that a party knows $\matr{V}$, $\matr{\Theta}$, $\matr{\Sigma}$ (all are in $\mathds{R}^{d \times d}$) and a vector $\vec{y}^* = (\matr{U}^{T} \vec{y})_{[1:d]} \in \mathds{R}^{d \times 1}$ such that:
 \begin{enumerate}
     \item \label{step:proveA} $\matr{A} = \matr{V} \matr{\Theta} \matr{V}^{T}$,
     \item \label{step:proveB} $\vec{b} = \matr{V} \matr{\Sigma}^T \vec{y}^*$,
     \item \label{step:proveV} $\matr{V}$ is an orthogonal matrix, namely, $\matr{V}^T \matr{V} = \matr{I}$, and
     \item \label{step:proveTheta} $\matr{\Theta}$ is a diagonal matrix where the diagonal entries are $1/{(\sigma_i^2 + \rho)}$. $\sigma_i$ are the values on the diagonal of $\matr{\Sigma}$ and $\rho$ is a public value.
 \end{enumerate}
 
 
Note that $\matr{\Gamma}$ can be readily derived from $\matr{\Sigma}$ by adding rows of zeros. 
Moreover, both $\matr{\Theta}$ and $\matr{\Sigma}$ are diagonal matrices.
Therefore, we only commit to the diagonal entries of $\matr{\Theta}$ and $\matr{\Sigma}$ since the rest of the entries are zeros.

The above four statements are sufficient to prove the properties of $\matr{A}$ and $\vec{b}$ in the new formulation.
The first two statements simply prove that $\matr{A}$ and $\vec{b}$ are indeed decomposed into \emph{some} matrices $\matr{V}$, $\matr{\Theta}$, $\matr{\Sigma}$, and $\matr{y}^{*}$.
Statement~\ref{step:proveV}) shows that $\matr{V}$ is an orthogonal matrix, since by definition an orthogonal matrix $\matr{Q}$ has to satisfy the equation $\matr{Q}^{T} \matr{Q} = \matr{I}$.
However, we  allow the prover to choose $\matr{V}$.
As stated before, the prover would have been free to choose $\matr{X}$ and $\vec{y}$ anyway, so this freedom does not give more power to the prover.

Statement~\ref{step:proveTheta}) proves that the matrix $\matr{\Theta}$ is a diagonal matrix such that the diagonal values satisfy the form above.
This is sufficient to show that $\matr{\Theta}$ is correct according to \emph{some} $\matr{\Sigma}$.
Again, the prover is free to choose $\matr{\Sigma}$, which is the same as freely choosing its input $\matr{X}$.

Finally, we chose to commit to $\vec{y}^{*}$ instead of committing to $\matr{U}$ and $\vec{y}$ separately.
Following our logic above, it seems that we also need to commit to $\matr{U}$ and prove that it is an orthogonal matrix, similar to what we did with $\matr{V}$.
This is not necessary because of an important property of orthogonal matrices: $\matr{U}$'s columns span the vector space $\mathds{R}^{n}$.
Multiplying $\matr{U} \vec{y}$, the result is a linear combination of the columns of $\matr{U}$.
Since we also allow the prover to pick its $\vec{y}$, $\matr{U}\vec{y}$ essentially can be any vector in $\mathds{R}^{n}$.
Thus, we only have to allow the prover to commit to the product of $\matr{U}$ and $\vec{y}$.
As we can see from the derivation, $\vec{b} = \matr{V} \matr{\Gamma}^{T} \matr{U} \vec{y}$, but since $\matr{\Gamma}$ is simply $\matr{\Sigma}$ with rows of zeros, the actual decomposition only needs the first $d$ elements of $\matr{U}\vec{y}$.
Hence, this allows us to commit to $\vec{y}^{*}$, which is $d \times 1$.

Using our techniques, \sys commits only to  matrices of sizes $d \times d$ or $d \times 1$, thus removing any scaling in $n$ (the number of rows in the dataset) in the input preparation phase.


\subsubsection{Proving the initial data summaries}

First, each party broadcasts $\enc{\matr{V}}$, $\enc{\matr{\Sigma}}$, $\enc{\matr{\Theta}}$, $\enc{\vec{y}^{*}}$, $\enc{\matr{A}}$, and $\enc{\vec{b}}$.
To encrypt a matrix, the party simply individually encrypts each entry.
The encryption scheme itself also acts as a commitment scheme~\cite{groth2009homomorphic}, so we do not need an extra commitment scheme.

To prove these statements, we also need another primitive called an interval proof.
Moreover, since these matrices act as inputs to the model compute phase, we also need to prove that $\matr{A}$ and $\vec{b}$ are within a certain range (this will be used by~\cref{gadget:spdz:shares}, described in~\cref{sec:model:release}).
The interval proof we use is from~\cite{boudot2000efficient}, which is an efficient way of proving that a committed number lies within a certain interval.
However, what we want to prove is that an encrypted number lies within a certain interval.
This can be solved by using techniques from~\cite{damgaard2002client}, which appends the range proof with a commitment-ciphertext equality proof.
This extra proof proves that, given a commitment and a Paillier ciphertext, both hide the same plaintext value.

To prove the first two statements, we invoke \cref{gadget:knowsone} and \cref{gadget:knowall}.
This allows us to prove that the party knows all of the matrices in question and that they satisfy the relations laid out in those statements.

There are two steps to proving statement~\ref{step:proveV}. 
The prover will compute $\enc{\matr{V}^T \matr{V}}$ and prove it computed it correctly using Gadget~\ref{gadget:knowsone} as above. 
The result should be equal to the encryption of the identity matrix.
However, since we are using fixed point representation for our data, the resulting matrix could be off from the expected values by some small error.
$\matr{V}^T \matr{V}$ will only be close to $\matr{I}$, but not equal to $\matr{I}$.
Therefore, we also utilize interval proofs to make sure that $\matr{V}^T \matr{V}$ is close to $\matr{I}$, without explicitly revealing the value of $\matr{V}^T \matr{V}$.

Finally, to prove statement~\ref{step:proveTheta}, the prover does the following:
\begin{enumerate}
\item The prover computes and releases $\enc{\matr{\Sigma}^2}$ because the prover knows $\matr{\Sigma}$ and proves using Gadget~\ref{gadget:knowsone} that this computation is done correctly.
\item The prover computes $\enc{\matr{\Sigma}^2 + \rho \matr{I}}$, which anyone can compute because $\rho$ and $\matr{I}$ are public. $\enc{\matr{\Sigma}^{2}}$ and $\enc{\rho \matr{I}}$ can be multiplied together to get the summation of the plaintext matrices.
\item The prover now computes $\enc{\matr{\Sigma}^2 + \rho \matr{I}} \times \enc{\matr{\Theta}}$ and proves this encryption was computed correctly using Gadget~\ref{gadget:knowsone}.
\item Similar to step 3), the prover ends this step by using interval proofs to prove that this encryption is close to encryption of the identity matrix. 
\end{enumerate}


\section{Model compute phase}
\label{sec:online:phase}

\subsection{Overview}
In the model compute phase, all parties use the summaries computed in the input preparation phase and execute the iterative ADMM training protocol.
An encrypted weight vector is generated at the end of this phase and distributed to all participants.
The participants can jointly decrypt this weight vector to get the plaintext model parameters.
This phase executes in three steps: initialization, training (local optimization and coordination), and model release.

\subsection{Initialization}
We initialize the weights $\vec{w}_i^{0}, \vec{z}^{0}, \text{and } \vec{u}_i^{0}$.
There are two popular ways of initializing the weights.
The first way is to set every entry to a random number.
The second way is to initialize every entry to zero.
In \sys, we use the second method because it is easy and works well in practice.

\subsection{Local optimization}
During ADMM's local optimization phase, each party takes the current weight vector and iteratively optimizes the weights based on its own dataset.
For LASSO, the update equation is simply $\vec{w}_i^{k+1} \leftarrow  \matr{A}_i (\vec{b}_i + \rho \left(\vec{z}^{k}  - \vec{u}_{i}^{k}\right))$,
where $\matr{A}_i$ is the matrix $(\matr{X}_i^T \matr{X}_i + \rho \matr{I})^{-1}$ and $\vec{b}_i$ is $\matr{X}_i^{T} \vec{y}_i$.
As we saw from the input preparation phase description, each party holds encryptions of $\matr{A}_i$ and $\vec{b}_i$.
Furthermore, given $\vec{z}^{k}$ and $\vec{u}^{k}_{i}$ (either initialized or received as results calculated from the previous round), each party can independently calculate $\vec{w}^{k+1}_{i}$ by doing plaintext scaling and plaintext-ciphertext matrix multiplication. 
Since this is done locally, each party also needs to generate a proof proving that the party calculated  $\vec{w}^{k+1}_{i}$  correctly.
We compute the proof for this step by invoking~\cref{gadget:knowsone}.

\subsection{Coordination using MPC}

After the local optimization step, each party holds encrypted weights $\vec{w}^{k+1}_{i}$.
The next step in the ADMM iterative optimization is the coordination phase.
Since this step contains non-linear functions, we evaluate it using generic MPC.

\subsubsection{Conversion to MPC}
First, the encrypted weights need to be converted into an MPC-compatible input.
To do so, we formulate a gadget that converts ciphertext to arithmetic shares.
The general idea behind the protocol is inspired by arithmetic sharing protocols~\cite{cramer2001multiparty,spdz}.

\newgadget{\label{gadget:additive:sharing}
For $m$ parties, each party having the public key $\pk$ and a share of the secret key $\sk$, given public ciphertext $\enc{a}$, convert $a$ into $m$ shares $a_i \in \mathds{Z}_{p}$ such that $a \equiv \sum a_i \mod p$.
Each party $P_i$ receives secret share $a_i$ and does not learn the original secret value $a$.
}

\parhead{Gadget usage}
Each party uses this gadget to convert $\enc{\vec{w}_i}$ and $\enc{\vec{u}_i}$ into input shares and compute the soft threshold function using MPC (in our case, SPDZ).
We denote $p$ as the public modulus used by SPDZ.
Note that all of the computation encrypted by ciphertexts are dong within modulo $p$.

\parhead{Protocol} The protocol proceeds as follows:
\begin{enumerate}
    \item Each party $P_i$ generates a random value $r_i \in [0, 2^{|p| + \kappa}]$ and encrypts it, where $\kappa$ is a statistical security parameter.
    Each party should also generate an interval plaintext proof of knowledge of $r_i$, then publish $\enc{r_i}$ along with the proofs.
    \item Each party $P_i$ takes as input the published $\{ \enc{r_j} \}_{j=1}^m$ and compute the product with $\enc{a}$. The result is $c = \enc{a + \sum_{j=1}^m r_j}$.
    \item All parties jointly decrypt $c$ to get plaintext $b$.
    \item Party $0$ sets $a_0 = b - r_0 \mod p$. Every other party sets $a_i \equiv - r_i \mod p$.
    \item Each party publishes $\enc{a_i}$ as well as an interval proof of plaintext knowledge.
\end{enumerate}

\subsubsection{Coordination}
The ADMM coordination step takes in $\vec{w}_i^{k+1}$ and $\vec{u}_i^{k}$, and outputs $\vec{z}^{k+1}$.
The $\vec{z}$ update requires computing the soft-threshold function (a non-linear function), so we express it in MPC.
Additionally, since we are doing fixed point integer arithmetic as well as using a relatively small prime modulus for MPC (256 bits in our implementation), we need to reduce the scaling factors accumulated on $\vec{w}^{k+1}_i$ during plaintext-ciphertext matrix multiplication.
We currently perform this operation inside MPC as well.

\subsubsection{Conversion from MPC}
After the MPC computation, each party receives shares of $\vec{z}$ and its MAC shares, as well as shares of $\vec{w}_i$ and its MAC shares.
It is easy to convert these shares back into encrypted form simply by encrypting the shares, publishing them, and summing up the encrypted shares. 
We can also calculate $\vec{u}_i^{k+1}$ this way.
Each party also publishes interval proofs of knowledge for each published encrypted cipher.
Finally, in order to verify that they are indeed valid SPDZ shares (the specific protocol is explained in the next section), each party also publishes  encryptions and interval proofs of all the MACs.

\subsection{Model release}
\label{sec:model:release}

\subsubsection{MPC conversion verification}
Since we are combining two protocols (homomorphic encryption and MPC), an attacker can attempt to alter the inputs to either protocol by using different or inconsistent attacker-chosen inputs. 
Therefore, before releasing the model, the parties must prove that they correctly executed the ciphertext to MPC conversion (and vice versa).
We use another gadget to achieve this.

\newgadget{\label{gadget:spdz:shares}
Given public parameters: encrypted value $\enc{a}$, encrypted SPDZ input shares $\enc{b_i}$, encrypted SPDZ MACs $\enc{c_i}$, and interval proofs of plaintext knowledge, verify that 
\begin{enumerate}[itemsep=0mm,leftmargin=*,topsep=0.0ex,parsep=0.0pt]
\item $a \equiv \sum_i b_i \mod p$, and 
\item $b_i$ are valid SPDZ shares and $c_i$'s are valid MACs on $b_i$.
\end{enumerate}
}
\parhead{Gadget usage}
We apply~\cref{gadget:spdz:shares} to all data that needs to be converted from encrypted ciphers to SPDZ or vice versa.
More specifically, we need to prove that (1) the SPDZ input shares are consistent with $\enc{\vec{w}_i^{k+1}}$ that is published from each party, and (2) the SPDZ shares for $\vec{w}_i^{k+1}$ and $\vec{z}^{k}$ are authenticated by the MACs.

\parhead{Protocol} The gadget construction proceeds as follows:
\begin{enumerate}
    \item Each party verifies that $\enc{a}$, $\enc{b_i}$ and $\enc{c_i}$ pass the interval proofs of knowledge. For example, $b_i$ and $c_i$ need to be within $[0, p]$.
    \item Each party homomorphically computes $\enc{\sum_i b_i}$, as well as $E_d = \enc{a - \sum_i b_i}$. \label{step:startcheck}
    \item Each party randomly chooses $r_i \in [0, 2^{|a| + |\kappa|}]$, where $\kappa$ is again a statistical security parameter, and publishes $\enc{r_i}$ as well as an interval proof of plaintext knowledge.\label{step:chooseri}
    \item Each party calculates $E_f = E_d \prod_i \enc{r_i}^{p} = \enc{(a - \sum_i b_i) + \sum_i (r_i \cdot p)}$. Here we assume that $\log |m| + |p| + |a| + |\kappa| < |n|$.
    \item All parties participate in a joint decryption protocol to decrypt $E_f$ obtaining $e_f$. 
    \item Every party individually checks to see that $e_f$ is a multiple of $p$. If this is not the case, abort the protocol. \label{step:endcheck}
    \item The parties release the SPDZ global MAC key $\alpha$.
    \item Each party calculates $\enc{\alpha (\sum b_i + \delta)}$ and $\enc{\sum c_i}$.
    \item Use the same method in steps \ref{step:startcheck} -- \ref{step:endcheck} to prove that $\alpha (\sum b_i + \delta) \equiv \sum c_i \mod p$.
\end{enumerate}

The above protocol is a way for parties to verify two things.
First, that the SPDZ shares indeed match with a previously published encrypted value (i.e., \cref{gadget:additive:sharing} was executed correctly).
Second, that the shares are valid SPDZ shares.
The second step is simply verifying the original SPDZ relation among value share, MAC shares, and the global key.

Note that we cannot verify these relations by simply releasing the plaintext data shares and their MACs since the data shares correspond to the intermediate weights.
Furthermore, the shares need to be equivalent in modulo $p$, which is different from the Paillier parameter $N$.
Therefore, we use an alternative protocol to test modulo equality between two ciphertexts, which is the procedure described above in steps \ref{step:startcheck} to \ref{step:endcheck}.

Since the encrypted ciphers come with interval proofs of plaintext knowledge, we can assume that  $a \in [0, l]$.
If two ciphertexts encrypt plaintexts that are equivalent to each other, they must satisfy that $a \equiv b \mod p$ or $a = b + \eta p$.
Thus, if we take the difference of the two ciphertexts, this difference must be $\eta p$.
We could then run the decryption protocol to test that the difference is indeed a multiple of $p$.

If $a \equiv \sum_i b_i \mod p$, simply releasing the difference could still reveal extra information about the value of $a$.
Therefore, all parties must each add a random mask to $a$.
In step~\ref{step:chooseri}, $r_i$'s are generated independently by all parties, which means that there must be at least one honest party who is indeed generating a random number within the range.
The resulting plaintext thus statistically hides the true value of $a - \sum_i b_i$ with the statistical parameter $\kappa$.
If $a \not\equiv \sum_i b_i \mod p$, then the protocol reveals the difference between $a - \sum_i b_i \mod p$.
This is safe because the only way to reveal $a - \sum_i b_i \mod p$ is when an adversary misbehaves and alters its inputs, and the result is independent from the honest party's behavior.


\subsubsection{Weight vector decryption}
Once all SPDZ values are verified, all parties jointly decrypt $\vec{z}$.
This can be done by first aggregating  the encrypted shares of $\vec{z}$ into a single ciphertext.
After this is done, the parties run the joint decryption protocol from~\cite{fouque2000sharing} (without releasing the private keys from every party).
The decrypted final weights are released in plaintext to everyone.


\section{Extensions to Other Models}
\label{sec:applications}
Though we used LASSO as a running example, our techniques can be applied to other linear models like ordinary least-squares linear regression, ridge regression, and elastic net.
Here we show the update rules for ridge regression, and leave its derivation to the readers.

Ridge regression solves a similar problem as LASSO, except with $L^{2}$ regularization.
Given a dataset $(\matr{X}, \vec{y})$ where $\matr{X}$ is the feature matrix and $\vec{y}$ is the prediction vector, ridge regression optimizes 
$\argmin_{\vec{w}} \dfrac{1}{2} \|\matr{X} \vec{w} - \vec{y}\|_{2}^{2} + \lambda \|\vec{w}\|_{2}$.
The update equations for ridge regression are:
\begin{align*}
\vec{w}_{i}^{k+1} &= (\matr{X}_i^{T} \matr{X}_i + \rho I)^{-1}(\matr{X}_i^{T}\vec{y}_i + \rho(\vec{z}^{k} - \vec{u}_i^{k})) \\
&+ (\rho/2) \| \vec{w}_i - \vec{z}^{k} + \vec{u}_i^{k} \|_{2}^{2}\\
\vec{z}^{k+1} &= \dfrac{\rho}{2\lambda/m + \rho} (\vec{\bar{w}}^{k+1} + \vec{\bar{u}}^{k})\\
\vec{u}_i^{k+1} &= \vec{u}_i^{k} + \vec{x}_i^{k+1} - \vec{z}^{k+1}
\end{align*}

The local update is similar to LASSO, while the coordination  update is a linear operation instead of the soft threshold function. 
Elastic net, which combines $L^1$ and $L^2$ regularization, can therefore be implemented by combining the regularization terms from LASSO and ridge regression.


\section{Evaluation}

We implemented \sys in C++.
We utilize the SPDZ library~\cite{SPDZgithub}, a mature library for maliciously secure multi-party computation, for both the baseline and \sys.
In our implementation, we apply the Fiat-Shamir heuristic to our zero-knowledge proofs~\cite{faust2012non}.
This technique is commonly used in implementations because it makes the protocols non-interactive and thus more efficient, but assumes the random oracle model. 

We compare \sys's performance to a maliciously secure baseline that trains using stochastic gradient descent, similar to SecureML~\cite{secureml}.
Since SecureML only supports two parties in the semihonest setting, we implemented a similar baseline using SPDZ~\cite{spdz}.
SecureML had a number of optimizations, but they were designed for the two-party setting.
We did not extend those optimizations to the multi-party setting.
We will refer to SGD implemented in SPDZ as the ``secure baseline'' (we explain more about the SGD training process in \cref{ssec:experiment:setup}).
Finally, we do not benchmark \sys's Paillier key setup phase.
This can be computed using SPDZ itself, and it is ran only once (as long as the party configuration does not change).

\subsection{Experiment setup}
\label{ssec:experiment:setup}
We ran our experiments on EC2 using r4.8xlarge instances. 
Each machine has 32 cores and 244 GiB of memory. 
In order to simulate a wide area network setting, we created EC2 instances in Oregon and Northern Virginia.
The instances are equally split across these two regions.
To evaluate \sys's scalability, we used synthetic datasets that are constructed by drawing samples from a noisy normal distribution.
For these datasets, we varied both the dimension and the number of parties. 
To evaluate \sys's performance against the secure baseline, we benchmarked both systems on two real world datasets from UCI~\cite{ucimlrepo}.

\parhead{Training assumptions.}
We do not tackle hyperparameter tuning in our work, and also assume that the data has been normalized before training.
We also use a fixed number of rounds ($10$) for ADMM training, which we found experimentally using the real world datasets.
We found that $10$ rounds is often enough for the training process to converge to a reasonable error rate.
Recall that ADMM converges in a small number of rounds because it iterates on a summary of the \emph{entire dataset}.
In contrast, SGD iteratively scans data from all parties at least once in order to get an accurate representation of the underlying distributions.
This is especially important when certain features occur rarely in a dataset.
Since the dataset is very large, even one pass already results in many rounds.

\parhead{MPC configuration.}
As mentioned earlier, SPDZ has two phases of computation: an offline phase and an online phase.
The offline phase can run independently of the secure function, but the precomputed values cannot be reused across multiple online phases.
The SPDZ library provides several ways of benchmarking different offline phases, including MASCOT~\cite{mascot} and Overdrive~\cite{overdrive}.
We tested both schemes and found Overdrive to perform better over the wide area network.
Since these are for benchmarking purposes only, we decided to estimate the SPDZ offline phase by dividing the number of triplets needed for a circuit by the benchmarked throughput.
The rest of the evaluation section will use the estimated numbers for all SPDZ offline computation.
Since \sys uses parallelism, we also utilized parallelism in the SPDZ offline generation by matching the number of threads on each machine to the number of cores available.

On the other hand, the SPDZ online implementation is not parallelized because the API was insufficient to effectively express parallelism.
We note two points.
First, while parallelizing the SPDZ library will result in a faster baseline, \sys also utilizes SPDZ, so any improvement to SPDZ also carries over to \sys.
Second, as shown below, our evaluation shows that \sys still achieves significant performance gains over the baseline even if the online phase in the secure baseline is infinitely fast.

Finally, the parameters we use for \sys are: 128 bits for the secure baseline's SPDZ configuration, 256 bits for the \sys SPDZ configuration, and 4096 bits for \sys's Paillier ciphertext.

\subsection{Theoretic performance}
\begin{table}[h]
    \centering
    \begin{tabular}{|cc||c|}
        \hline
        {\bf Baseline} & Secure SGD & $ C \cdot m^2 \cdot n \cdot d$ \\
        \hline
        {\bf \sys} & SVD decomposition & $c_1 \cdot n \cdot d^2$\\
        & SVD proofs & $c_1 \cdot m \cdot d^2 + c_2 \cdot d^3$ \\
        & MPC offline &  $c_1 \cdot m^2 \cdot d$ \\
        & Model compute & $c_1 \cdot m^2 \cdot d + c_2 \cdot d^2 + c_3 \cdot m \cdot d$ \\
        \hline
    \end{tabular}
    \caption{Theoretical scaling (complexity analysis) for SGD baseline and \sys. $m$ is the number of parties, $n$ is the number of samples per party, $d$ is the dimension.
    }
    \label{table:scaling}
\end{table}

\cref{table:scaling} shows the theoretic scaling behavior for SGD and \sys, where $m$ is the number of parties, $n$ is the number of samples per party, $d$ is the dimension, and $C$ and $c_i$ are constants.
Note that $c_i$'s are not necessarily the same across the different rows in the table.
We split \sys's input preparation phase into three sub-components: SVD (calculated in plaintext), SVD proofs, and MPC offline (since \sys uses SPDZ during the model compute phase, we also need to run the SPDZ offline phase).

SGD scales linearly in $n$ and $d$.
If the number of samples per party is doubled, the number of iterations is also doubled.
A similar argument goes for $d$.
SGD scales quadratic in $m$ because it first scales linearly in $m$ due to the behavior of the MPC protocol.
If we add more parties to the computation, the number of samples will also increase, which in turn increases the number of iterations needed to scan the entire dataset.

\sys, on the other hand,  scales linearly in $n$ only for the SVD computation. 
We emphasize that SVD is very  fast because it is {\em executed on plaintext data}.  
The $c_1$ part of the SVD proofs formula scales linearly in $m$ because each party needs to verify from every other party. 
It also scales linearly in $d^2$ because each proof verification requires $d^2$ work.
The $c_2$ part of the formula has $d^3$ scaling because our matrices are $d \times d$), and to calculate a resulting encrypted matrix requires matrix multiplication on two $d \times d$ matrices.

The coordination phase from \sys's model compute phase, as well as the corresponding MPC offline compute phase, scale quadratic in $m$ because we need to use MPC to re-scale weight vectors from each party.
This cost corresponds to the $c_1$ part of the formula.
The model compute phase's $d^{2}$ cost ($c_2$ part of the formula) reflects the matrix multiplication and the proofs.
The rest of the MPC conversion proofs scale linearly in $m$ and $d$ ($c_3$ part of the formula).

\begin{figure}[t!]
    \centering 
    \begin{subfigure}[t]{0.45\textwidth}
        \includegraphics[width=\textwidth]{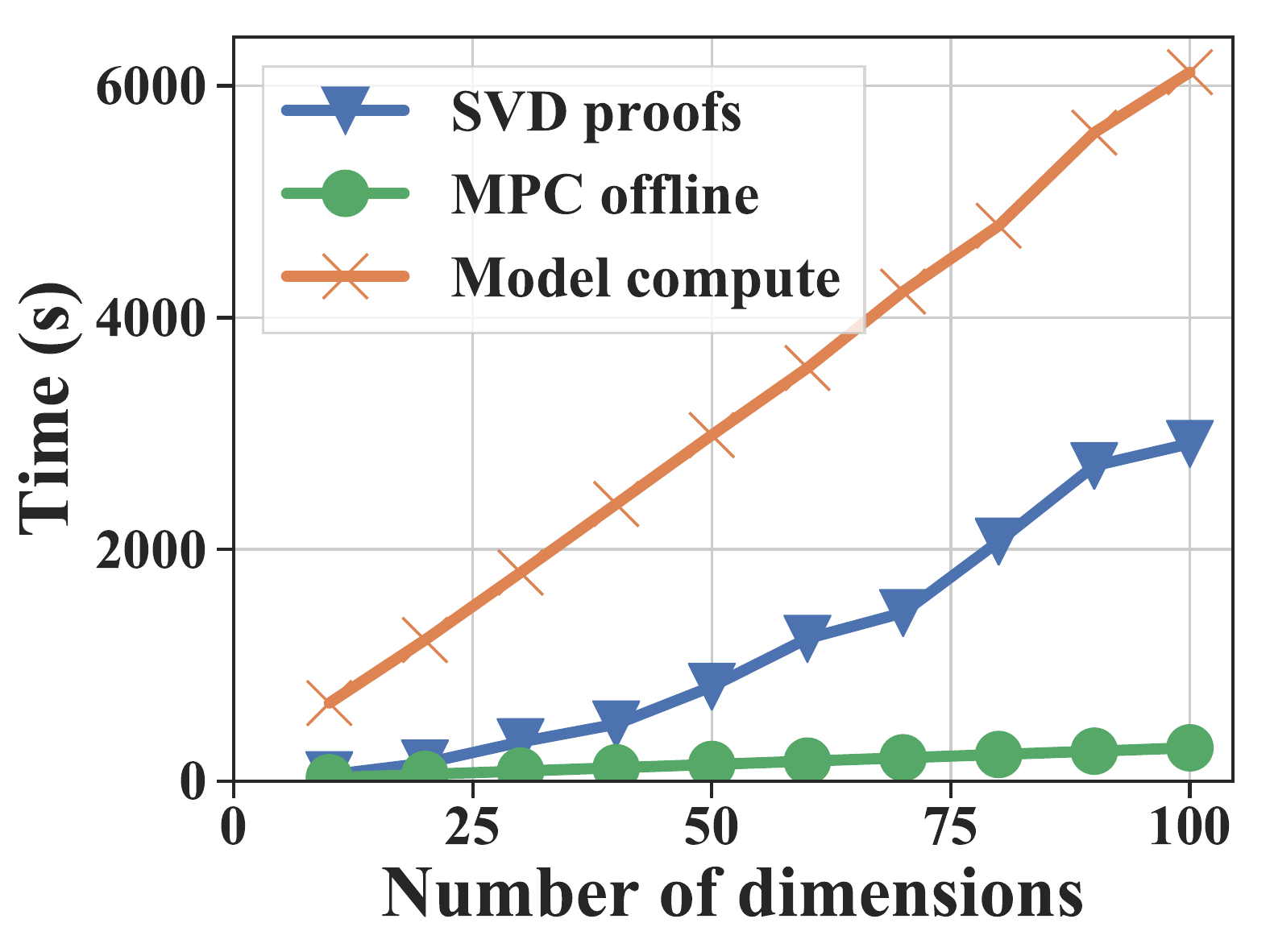}
        \caption{\sys's scaling as we increase the number of dimensions. The number of parties is fixed to be 4, and the number of samples per party is $100,000$.}
        \label{fig:vary:dim}
    \end{subfigure}
    \hfill
    \begin{subfigure}[t]{0.45\textwidth}
        \centering
        \includegraphics[width=\textwidth]{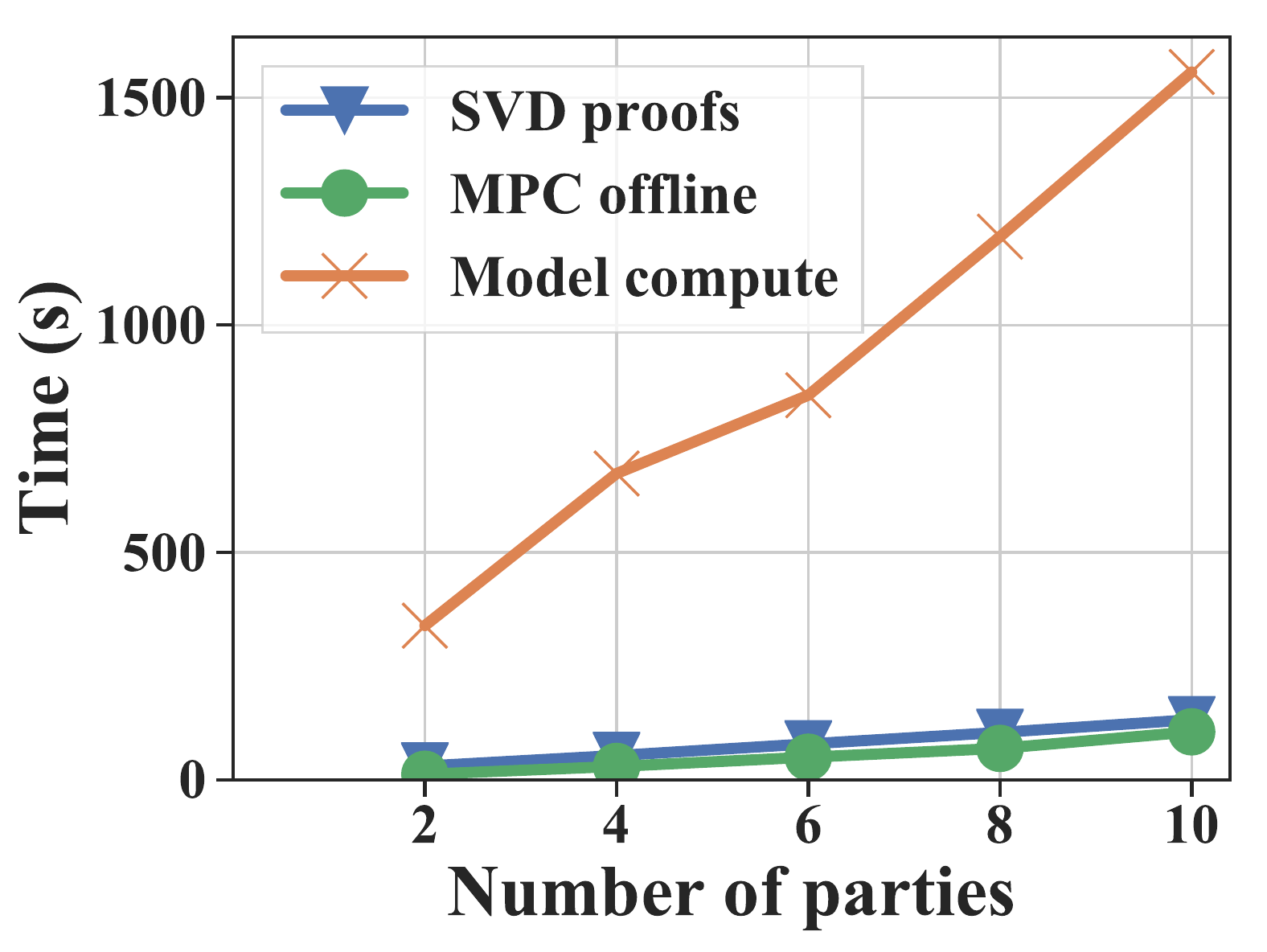}
        \caption{\sys's two phases as we increase the number of parties. The dimension is set to be $10$, and the number of samples per party is $100,000$.} 
        \label{fig:vary:parties}
    \end{subfigure}
    \caption{\sys scalability measurements.}
\end{figure}

\subsection{Synthetic datasets}
\label{sec:eval:synthetic}

\begin{table*}[t!]
    \centering
    \scriptsize
    \begin{tabular}{|c||c|c|c|c|c|c|c|c|c|c|c|}
        \hline
         Samples per party &  2000 & 4000 & 6000 & 8000 & 10K & 40K & 100K & 200K & 400K & 800K & 1M\\
         \hline
         sklearn L2 error & 8937.01 & 8928.32 & 8933.64 & 8932.97 & 8929.10 & 8974.15 & 8981.24 & 8984.64 & 8982.88 & 8981.11 & 8980.35\\
         \sys L2 error & 8841.33 & 8839.96 & 8828.18 & 8839.56 & 8837.59 & 8844.31 & 8876.00 & 8901.84 & 8907.38 & 8904.11 & 8900.37  \\
         \hline
         sklearn MAE &  57.89 & 58.07 & 58.04 & 58.10 & 58.05 & 58.34 & 58.48 & 58.55 & 58.58 & 58.56 & 58.57 \\
         \sys MAE & 57.23 & 57.44 & 57.46 & 57.44 & 57.47 & 57.63 & 58.25 & 58.38 & 58.36 & 58.37 & 58.40 \\
         \hline
    \end{tabular}
    \caption{Select errors for gas sensor (due to space), comparing \sys with a baseline that uses sklearn to train on all plaintext data. 
      L2 error is the squared norm; MAE is the mean average error.
    Errors are calculated after post-processing. }
    \label{table:gas_sensor_error_rates}
\end{table*}

\begin{table*}[t!]
    \centering
    \small
    \begin{tabular}{|c||c|c|c|c|c|c|c|c|c|c|c|}
        \hline
         Samples per party &  1000 & 2000 & 4000 & 6000 & 8000 & 10K & 20K & 40K & 60K & 80K & 100K\\
         \hline
         sklearn L2 error & 92.43 & 91.67 & 90.98 & 90.9 & 90.76 & 90.72 & 90.63 & 90.57 & 90.55 & 90.56 & 90.55\\
         \sys L2 error & 93.68 & 91.8 & 91.01 & 90.91 & 90.72 & 90.73 &  90.67 & 90.57 & 90.54 & 90.57 & 90.55\\
         \hline
         sklearn MAE & 6.86 & 6.81 & 6.77 & 6.78 & 6.79 & 6.81 & 6.80 & 6.79 & 6.79 & 6.80 & 6.80\\
         \sys MAE & 6.92 & 6.82 & 6.77 & 6.78 & 6.79 & 6.81 & 6.80 & 6.79 & 6.80 & 6.80 & 6.80 \\
         \hline
    \end{tabular}
    \caption{Errors for song prediction, comparing \sys with a baseline that uses sklearn to train on all plaintext data. 
      L2 error is the squared norm; MAE is the mean average error.
    Errors are calculated after post-processing.}
    \label{table:song_prediction_error_rates}
\end{table*}

We want to answer two questions about \sys's scalability using synthetic datasets: 
how does \sys scale as we vary the number of features 
and  how does it scale as we vary the number of parties?
Note that we are not varying the number of input samples because that will be explored in~\cref{sec:eval:real} in comparison to the secure SGD baseline.

\cref{fig:vary:dim} shows a breakdown of \sys's cryptographic computation as we scale the number of dimensions.
The plaintext SVD computation is not included in the graph.
The SVD proofs phase is dominated by the matrix multiplication proofs, which scales in $d^{2}$.
The MPC offline phase and the model compute phase are both dominated by the linear scaling in $d$, which corresponds to the MPC conversion proofs.

\cref{fig:vary:parties} shows the same three phases as we increase the number of parties.
The SVD proofs phase scales linearly in the number of parties $m$.
The MPC offline phase scales quadratic in $m$, but its effects are not very visible for a small number of parties.
The model compute phase is dominated by the linear scaling in $m$ because the quadratic scaling factor isn't very visible for a small number of parties.


Finally, we also ran a microbenchmark to understand \sys's network and compute costs.
The experiment used 4 servers and a synthetic dataset with 50 features and 100K samples per party.
We found that the network costs account for approximately 2\% of the input preparation phase and 22\% of \sys's model compute phase.

\subsection{Real world datasets}
\label{sec:eval:real}

We evaluate on two different real world datasets: gas sensor data~\cite{ucimlrepo} and the million song dataset~\cite{bertin2011million,ucimlrepo}.
The gas sensor dataset records 16 sensor readings when mixing two types of gases. 
Since the two gases are mixed with random concentration levels, the two regression variables are independent and we can simply run two different regression problems (one for each gas type).
For the purpose of benchmarking, we ran an experiment using the ethylene data in the first dataset.
The million song dataset is used for predicting a song's published year using 90 features.
Since regression problems produce real values, the year can be calculated by rounding the regressed value.

For SGD, we set the batch size to be the same size as the dimension of the dataset.
The number of iterations is equal to the total number of sample points divided by the batch size.
Unfortunately, we had to extrapolate the runtimes for a majority of the baseline online phases because the circuits were too big to compile on our EC2 instances.

\cref{fig:gas_sensor} and~\cref{fig:song_prediction} compare \sys to the baseline on the two datasets.
Note that \sys's input preparation graph combines the three phases that are run during the offline: plaintext SVD computation, SVD proofs, and MPC offline generation.
We can see that \sys's input preparation phase scales very slowly with the number of samples.
The scaling actually comes from the plaintext SVD calculation because both the SVD proofs and the MPC offline generation do not scale with the number of samples.
\sys's model compute phase also stays constant because we fixed the number of iterations to a conservative estimate.
SGD, on the other hand, does scale linearly with the number of samples in both the offline and the online phases.

\begin{figure}[t]
    \centering
    \includegraphics[width=0.45\textwidth]{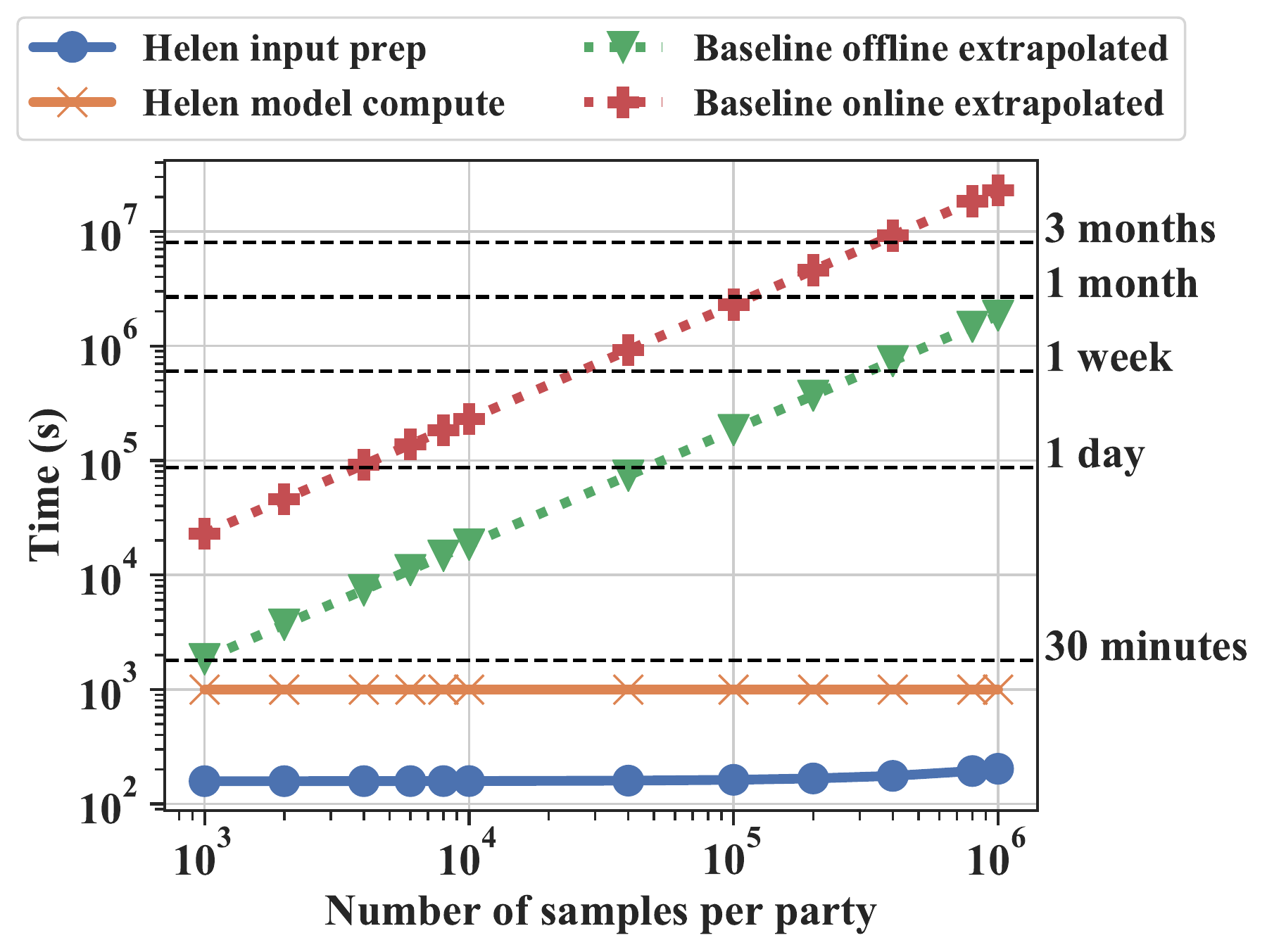}
    \caption{\sys and baseline performance on the gas sensor data. The gas sensor data contained over 4 million data points; we partitioned into 4 partitions with varying number of sample points per partition to simulate the varying  number of samples per party. 
    The number of parties is 4, and the number of dimensions is $16$.}
    \label{fig:gas_sensor}
\end{figure}

\begin{figure}[t]
    \centering
    \includegraphics[width=0.45\textwidth]{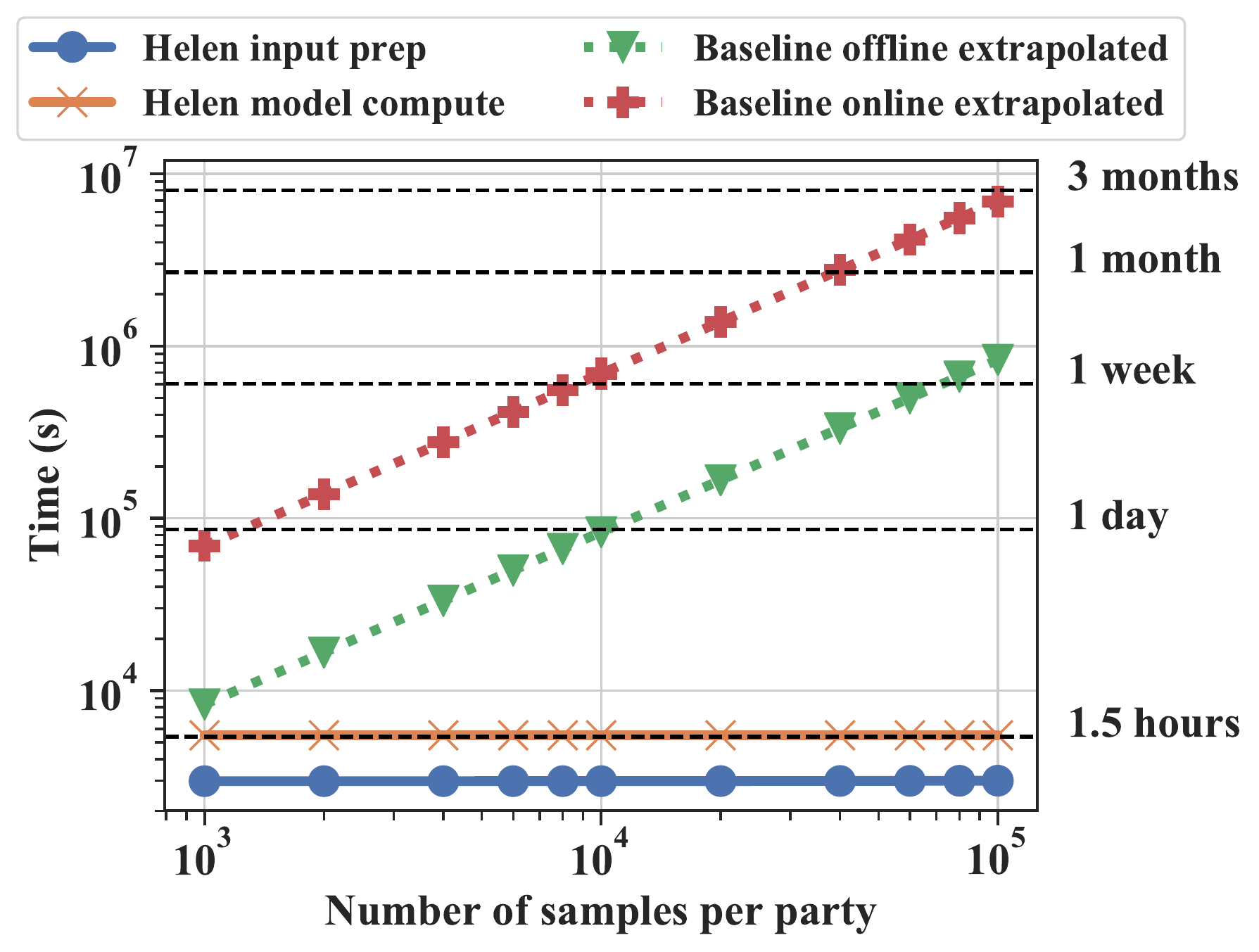}
    \caption{\sys and baseline performance on the song prediction data, as we vary the number of samples per party. The number of parties is 4, and the number of dimensions is $90$.
    }
    \label{fig:song_prediction}
    \caption{\sys comparison with SGD}
\end{figure}

For the gas sensor dataset, \sys's total runtime (input preparation plus model compute) is able to achieve a 21.5x performance gain over the baseline's total runtime (offline plus online) when the number of samples is 1000.
When the number of samples per party reaches 1 million, \sys is able to improve over the baseline by 20689x.
For the song prediction dataset, \sys is able to have a 9.1x performance gain over the baseline when the number of samples is 1000.
When the number of samples per party reaches 100K, \sys improves over the baseline by 911x.
Even if we compare \sys to the baseline's offline phase only, we find that \sys still has close to constant scaling while the baseline scales linearly with the number of samples.
The performance improvement compared to the baseline offline phase is up to 1540x for the gas sensor dataset and up to 98x for the song prediction dataset.

In \cref{table:gas_sensor_error_rates} and \cref{table:song_prediction_error_rates}, we evaluate \sys's test errors on the two datasets.
We compare the L2 and mean average error for \sys to the errors obtained from a model trained using sklearn (a standard Python library for machine learning) on the plaintext data.
We did not directly use the SGD baseline because its online phase does not compile for larger instances, and using sklearn on the plaintext data is a conservative estimate.
We can see that \sys achieves similar errors compared to the sklearn baseline.









\newcommand{\badcell}{\cellcolor{red!25}}
\begin{figure*}[t]
\centering
\begin{tabular}{p{4.5cm}|p{3cm}|p{2cm}|p{3cm}|p{1cm}}
{\bf Work} & {\bf Functionality} & {\bf n-party?} & {\bf Malicious security?} & {\bf Practical?} \\ 
\hline
Nikolaenko et al.~\cite{Nikolaenko} & ridge regression & \badcell{} no & \badcell{} no & -- \\
Hall et al.~\cite{HallRidge} & linear regression  & {\bf yes}  & \badcell{} no &  -- \\ 
Gascon et al.~\cite{DistributedGascon} & linear regression & \badcell{} no & \badcell{} no & -- \\
Cock et al.~\cite{Cock2015} & linear regression & \badcell{} no & \badcell{} no & -- \\
Giacomelli et al.~\cite{cryptoeprint:2017:979} & ridge regression & \badcell{} no & \badcell{} no & -- \\
Alexandru et al.~\cite{quadraticopt} & quadratic opt. & \badcell{} no & \badcell no & -- \\
SecureML~\cite{secureml} & linear, logistic, deep learning & \badcell{} no & \badcell{} no & -- \\
Shokri\&Shmatikov~\cite{ShokriShmatikov} & deep learning & \badcell{} not MPC (heuristic) & \badcell{} no & -- \\
Semi-honest MPC~\cite{FairplayMP} & any function & {\bf yes} & \badcell{}  no & -- \\
Malicious MPC~\cite{spdz,Goldreich:1987,Sharemind,VIFF} & any function & {\bf yes} & {\bf yes} & \badcell{} no \\
\hline
\multicolumn{2}{p{8.0cm}|}{{\bf Our proposal, \sys}: regularized linear models} & {\bf yes} & {\bf yes} & {\bf yes} \\ 
\end{tabular}
\caption{ {\bf Insufficiency of existing cryptographic approaches.} ``n-party'' refers to whether the $n$($>$2) organizations can perform the computation with {\em equal trust} (thus not including the two non-colluding servers model).
We answer the practicality question only for maliciously-secure systems. 
We note that a few works that we marked as not coopetitive and not maliciously secure discuss at a high level how one might extend their work to such a setting, but they did not flesh out designs or evaluate their proposals.}
\label{table:related:work}
\end{figure*}

\section{Related work}\label{sec:related}

We organize the related work section into related coopetitive systems and attacks.

\subsection{Coopetitive systems}

\parhead{Coopetitive training systems}

In~\cref{table:related:work}, we  compare \sys to prior coopetitive training systems~\cite{Nikolaenko,HallRidge,DistributedGascon,Cock2015,cryptoeprint:2017:979,quadraticopt,secureml,ShokriShmatikov}.
The main takeaway is that, excluding {\em generic} maliciously secure MPC,  prior training systems do not provide malicious security.
Furthermore, most of them also assume that the training process requires outsourcing to two non-colluding servers.
At the same time, and as a result of choosing a weaker security model, some of these systems provide richer functionality than \sys, such as support for neural networks.
As part of our future work, we are exploring how to apply \sys's techniques to logistic regression and neural networks.

\parhead{Other coopetitive systems}
Other than coopetitive training systems, there are prior works on building coopetitive systems for applications like machine learning prediction and SQL analytics.
Coopetitive prediction systems~\cite{MLEnc, DeepSecure,chameleon,MiniONN,CryptoNets,Gazelle} typically consist of two parties, where one party holds a model and the other party holds an input. 
The two parties jointly compute a prediction without revealing the input or the model to the other party.
Coopetitive analytics systems~\cite{SMCQL,DJoin,Bonawitz17,Prio,Prochlo} allow multiple parties to run SQL queries over all parties' data.
These computation frameworks do not directly translate to \sys's training workloads.
Most of these works also do not address the malicious setting.
Recent work has also explored secure learning and analytics using separate compute nodes and blockchains~\cite{froelicher2019drynx,froelicher2017unlynx}.
The setup is different from that of \sys where we assume that the data providers are malicious and are also performing and verifying the computation.

\parhead{Trusted hardware based systems}
The related work presented in the previous two sections all utilize purely software based solutions.
Another possible approach is to use trusted hardware~\cite{sgx,costan2016intel}, and there are various secure distributed systems that could be extended to the coopetitive setting~\cite{vc3,ryoan,opaque}.
However, these hardware mechanisms require additional trust and are prone to side-channel leakages~\cite{spectre,van2018foreshadow,lee2017inferring}. 

\subsection{Attacks on machine learning}
\label{sec:ml:attacks}
Machine learning attacks can be categorized into data poisoning, model leakage, parameter stealing, and adversarial learning.
As mentioned in \S\ref{s:threat}, \sys tackles the problem of cryptographically running the training algorithm without sharing datasets amongst the parties involved, while defenses against these attacks are  {\em orthogonal} and complementary to our goal in this paper.
Often, these machine learning attacks can be separately addressed outside of \sys.
We briefly discuss two relevant attacks related to the training stage and some methods for mitigating them.

\parhead{Poisoning}
Data poisoning allows an attacker to inject poisoned inputs into a dataset before training~\cite{jagielski2018manipulating,chen2017targeted}.
Generally, malicious MPC does not prevent an attacker from choosing incorrect initial inputs because there is no way to enforce this requirement.
Nevertheless, there are some ways of mitigating arbitrary poisoning of data that would complement \sys's training approach.
Before training, one can  check that the inputs are confined within certain intervals.
The training process itself can also execute \emph{cross validation}, a process that can identify parties that do not contribute useful data.
After training, it is possible to further post process the model via techniques like fine tuning and parameter pruning~\cite{liu2018fine}.

\parhead{Model leakage}
Model leakage~\cite{shmatikovmachine,carlini2018secret} is an attack launched by an adversary who tries to infer information about the training data from the model itself.
Again, malicious MPC does not prevent an attacker from learning the final result.
In our coopetitive model, we also assume that all parties want to cooperate and have agreed to release the final model to everyone.

\subsection{Differential privacy}
\label{sec:diffpriv}
One way to alleviate model leakage is through the use of differential privacy~\cite{iyengartowards,abadi2016deep,duchi2013local}.
For example, one way to add differential privacy is to add carefully chosen noise directly to the output model~\cite{iyengartowards}.
Each party's noise can be chosen directly using MPC, and the final result can then be added to the final model before releasing.
In \sys, differential privacy would be added after the model is computed, but before the model release phase. 
However, there are more complex techniques for differential privacy that involve modification to the training algorithm, and integrating this into~\sys is an interesting future direction to explore.





\section{Conclusion}
In this paper, we propose \sys, a coopetitive system for training linear models.
Compared to prior work, \sys assumes a stronger threat model by defending against {\em malicious} participants.
This means that each party only needs to trust itself.
Compared to a baseline implemented with a state-of-the-art malicious framework, \sys is able to achieve up to five orders of magnitude of performance improvement.
Given the lack of efficient maliciously secure training protocols, we hope that our work on \sys will lead to further work on efficient systems with such strong security guarantees. 


\section{Acknowledgment}
We thank the anonymous reviewers for their valuable reviews, as well as Shivaram Venkataraman, Stephen Tu, and Akshayaram Srinivasan for their feedback and discussions.
This research was supported by NSF CISE Expeditions Award CCF-1730628, as well as gifts from the Sloan Foundation, Hellman Fellows Fund, Alibaba, Amazon Web Services, Ant Financial, Arm, Capital One, Ericsson, Facebook, Google, Huawei, Intel, Microsoft, Scotiabank, Splunk and VMware.

{\footnotesize \bibliographystyle{acm}
\bibliography{ref}}

\begin{thebibliography}{10}

\bibitem{SPDZgithub}
bristolcrypto/spdz-2: Multiparty computation with {SPDZ}, {MASCOT}, and
  {Overdrive} offline phases.
\newblock \url{https://github.com/bristolcrypto/SPDZ-2}.
\newblock Accessed: 2018-10-31.

\bibitem{VIFF}
{VIFF, the Virtual Ideal Functionality Framework}.
\newblock \url{http://viff.dk/}, 2015.

\bibitem{hipaa}
Health insurance portability and accountability act, April 2000.

\bibitem{abadi2016deep}
{\sc Abadi, M., Chu, A., Goodfellow, I., McMahan, H.~B., Mironov, I., Talwar,
  K., and Zhang, L.}
\newblock Deep learning with differential privacy.
\newblock In {\em Proceedings of the 2016 ACM SIGSAC Conference on Computer and
  Communications Security\/} (2016), ACM, pp.~308--318.

\bibitem{quadraticopt}
{\sc Alexandru, A.~B., Gatsis, K., Shoukry, Y., Seshia, S.~A., Tabuada, P., and
  Pappas, G.~J.}
\newblock Cloud-based quadratic optimization with partially homomorphic
  encryption.
\newblock {\em arXiv preprint arXiv:1809.02267\/} (2018).

\bibitem{SMCQL}
{\sc Bater, J., Elliott, G., Eggen, C., Goel, S., Kho, A., and Rogers, J.}
\newblock Smcql: secure querying for federated databases.
\newblock {\em Proceedings of the VLDB Endowment 10}, 6 (2017), 673--684.

\bibitem{FairplayMP}
{\sc Ben-David, A., Nisan, N., and Pinkas, B.}
\newblock Fairplaymp: a system for secure multi-party computation.
\newblock \url{www.cs.huji.ac.il/project/Fairplay/FairplayMP.html}, 2008.

\bibitem{ben1988completeness}
{\sc Ben-Or, M., Goldwasser, S., and Wigderson, A.}
\newblock Completeness theorems for non-cryptographic fault-tolerant
  distributed computation.
\newblock In {\em Proceedings of the twentieth annual ACM symposium on Theory
  of computing\/} (1988), ACM, pp.~1--10.

\bibitem{bertin2011million}
{\sc Bertin-Mahieux, T., Ellis, D.~P., Whitman, B., and Lamere, P.}
\newblock The million song dataset.
\newblock In {\em Ismir\/} (2011), vol.~2, p.~10.

\bibitem{Prochlo}
{\sc Bittau, A., Erlingsson, U., Maniatis, P., Mironov, I., Raghunathan, A.,
  Lie, D., Rudominer, M., Kode, U., Tinnes, J., and Seefeld, B.}
\newblock Prochlo: Strong privacy for analytics in the crowd.
\newblock In {\em Proceedings of the 26th Symposium on Operating Systems
  Principles\/} (2017), ACM, pp.~441--459.

\bibitem{Sharemind}
{\sc Bogdanov, D., Laur, S., and Willemson, J.}
\newblock {\em Sharemind: A Framework for Fast Privacy-Preserving
  Computations}.
\newblock 2008.

\bibitem{Bonawitz17}
{\sc Bonawitz, K., Ivanov, V., Kreuter, B., Marcedone, A., McMahan, H.~B.,
  Patel, S., Ramage, D., Segal, A., and Seth, K.}
\newblock Practical secure aggregation for privacy-preserving machine learning.
\newblock In {\em Proceedings of the 2017 ACM SIGSAC Conference on Computer and
  Communications Security\/} (2017), CCS '17.

\bibitem{MLEnc}
{\sc Bost, R., Popa, R.~A., Tu, S., and Goldwasser, S.}
\newblock Machine learning classification over encrypted data.
\newblock In {\em Network and Distributed System Security Symposium (NDSS)\/}
  (2015).

\bibitem{boudot2000efficient}
{\sc Boudot, F.}
\newblock Efficient proofs that a committed number lies in an interval.
\newblock In {\em International Conference on the Theory and Applications of
  Cryptographic Techniques\/} (2000), Springer, pp.~431--444.

\bibitem{ADMM}
{\sc Boyd, S., Parikh, N., Chu, E., Peleato, B., and Eckstein, J.}
\newblock Distributed optimization and statistical learning via the alternating
  direction method of multipliers.
\newblock In {\em Foundations and Trends in Machine Learning, Vol. 3, No. 1\/}
  (2010).

\bibitem{canetti2006security}
{\sc Canetti, R.}
\newblock Security and composition of cryptographic protocols: a tutorial (part
  i).
\newblock {\em ACM SIGACT News 37}, 3 (2006), 67--92.

\bibitem{carlini2018secret}
{\sc Carlini, N., Liu, C., Kos, J., Erlingsson, {\'U}., and Song, D.}
\newblock The secret sharer: Measuring unintended neural network memorization
  \& extracting secrets.
\newblock {\em arXiv preprint arXiv:1802.08232\/} (2018).

\bibitem{lassoForCreditScoring}
{\sc Chen, H., and Xiang, Y.}
\newblock The study of credit scoring model based on group lasso.
\newblock {\em Procedia Computer Science 122\/} (2017), 677 -- 684.
\newblock 5th International Conference on Information Technology and
  Quantitative Management, ITQM 2017.

\bibitem{chen2017targeted}
{\sc Chen, X., Liu, C., Li, B., Lu, K., and Song, D.}
\newblock Targeted backdoor attacks on deep learning systems using data
  poisoning.
\newblock {\em arXiv preprint arXiv:1712.05526\/} (2017).

\bibitem{mpcfairness}
{\sc Cleve, R.}
\newblock Limits on the security of coin flips when half the processors are
  faulty.
\newblock In {\em Proceedings of the eighteenth annual ACM symposium on Theory
  of computing\/} (1986), ACM, pp.~364--369.

\bibitem{Cock2015}
{\sc Cock, M.~d., Dowsley, R., Nascimento, A.~C., and Newman, S.~C.}
\newblock Fast, privacy preserving linear regression over distributed datasets
  based on pre-distributed data.
\newblock In {\em Proceedings of the 8th ACM Workshop on Artificial
  Intelligence and Security (AISec)\/} (2015).

\bibitem{Prio}
{\sc Corrigan-Gibbs, H., and Boneh, D.}
\newblock Prio: Private, robust, and scalable computation of aggregate
  statistics.
\newblock In {\em 14th {USENIX} Symposium on Networked Systems Design and
  Implementation ({NSDI} 17)\/} (2017).

\bibitem{costan2016intel}
{\sc Costan, V., and Devadas, S.}
\newblock Intel sgx explained.
\newblock {\em IACR Cryptology ePrint Archive 2016\/} (2016), 86.

\bibitem{cramer2001multiparty}
{\sc Cramer, R., Damg{\aa}rd, I., and Nielsen, J.}
\newblock Multiparty computation from threshold homomorphic encryption.
\newblock {\em EUROCRYPT 2001\/} (2001), 280--300.

\bibitem{damgaard2000efficient}
{\sc Damg{\aa}rd, I.}
\newblock Efficient concurrent zero-knowledge in the auxiliary string model.
\newblock In {\em International Conference on the Theory and Applications of
  Cryptographic Techniques\/} (2000), Springer, pp.~418--430.

\bibitem{damgaard2002sigma}
{\sc Damg{\aa}rd, I.}
\newblock On $\sigma$-protocols.
\newblock {\em Lecture Notes, University of Aarhus, Department for Computer
  Science\/} (2002).

\bibitem{damgaard2002client}
{\sc Damg{\aa}rd, I., and Jurik, M.}
\newblock Client/server tradeoffs for online elections.
\newblock In {\em International Workshop on Public Key Cryptography\/} (2002),
  Springer, pp.~125--140.

\bibitem{spdz}
{\sc Damg{\aa}rd, I., Pastro, V., Smart, N., and Zakarias, S.}
\newblock Multiparty computation from somewhat homomorphic encryption.
\newblock In {\em Advances in Cryptology--CRYPTO 2012}. Springer, 2012,
  pp.~643--662.

\bibitem{lassoForGenomics1}
{\sc D'Angelo, G.~M., Rao, D.~C., and Gu, C.~C.}
\newblock Combining least absolute shrinkage and selection operator (lasso) and
  principal-components analysis for detection of gene-gene interactions in
  genome-wide association studies.
\newblock In {\em BMC proceedings\/} (2009).

\bibitem{ucimlrepo}
{\sc Dheeru, D., and Karra~Taniskidou, E.}
\newblock {UCI} machine learning repository, 2017.

\bibitem{coopetition}
{\sc Dictionaries, E.~O.}
\newblock Coopetition.

\bibitem{duchi2013local}
{\sc Duchi, J.~C., Jordan, M.~I., and Wainwright, M.~J.}
\newblock Local privacy, data processing inequalities, and statistical minimax
  rates.
\newblock {\em arXiv preprint arXiv:1302.3203\/} (2013).

\bibitem{faust2012non}
{\sc Faust, S., Kohlweiss, M., Marson, G.~A., and Venturi, D.}
\newblock On the non-malleability of the fiat-shamir transform.
\newblock In {\em International Conference on Cryptology in India\/} (2012),
  Springer, pp.~60--79.

\bibitem{fouque2000sharing}
{\sc Fouque, P.-A., Poupard, G., and Stern, J.}
\newblock Sharing decryption in the context of voting or lotteries.
\newblock In {\em International Conference on Financial Cryptography\/} (2000),
  Springer, pp.~90--104.

\bibitem{froelicher2017unlynx}
{\sc Froelicher, D., Egger, P., Sousa, J.~S., Raisaro, J.~L., Huang, Z.,
  Mouchet, C., Ford, B., and Hubaux, J.-P.}
\newblock Unlynx: a decentralized system for privacy-conscious data sharing.
\newblock {\em Proceedings on Privacy Enhancing Technologies 2017}, 4 (2017),
  232--250.

\bibitem{froelicher2019drynx}
{\sc Froelicher, D., Troncoso-Pastoriza, J.~R., Sousa, J.~S., and Hubaux,
  J.-P.}
\newblock Drynx: Decentralized, secure, verifiable system for statistical
  queries and machine learning on distributed datasets.
\newblock {\em arXiv preprint arXiv:1902.03785\/} (2019).

\bibitem{garay2003strengthening}
{\sc Garay, J.~A., MacKenzie, P., and Yang, K.}
\newblock Strengthening zero-knowledge protocols using signatures.
\newblock In {\em Eurocrypt\/} (2003), vol.~2656, Springer, pp.~177--194.

\bibitem{DistributedGascon}
{\sc Gascón, A., Schoppmann, P., Balle, B., Raykova, M., Doerner, J., Zahur,
  S., and Evans, D.}
\newblock Privacy-preserving distributed linear regression on high-dimensional
  data.
\newblock Cryptology ePrint Archive, Report 2016/892, 2016.

\bibitem{cryptoeprint:2017:979}
{\sc Giacomelli, I., Jha, S., Joye, M., Page, C.~D., and Yoon, K.}
\newblock Privacy-preserving ridge regression with only linearly-homomorphic
  encryption.
\newblock Cryptology ePrint Archive, Report 2017/979, 2017.
\newblock \url{https://eprint.iacr.org/2017/979}.

\bibitem{CryptoNets}
{\sc Gilad-Bachrach, R., Dowlin, N., Laine, K., Lauter, K., Naehrig, M., and
  Wernsing, J.}
\newblock Cryptonets: Applying neural networks to encrypted data with high
  throughput and accuracy.
\newblock In {\em International Conference on Machine Learning\/} (2016),
  pp.~201--210.

\bibitem{Goldreich:1987}
{\sc Goldreich, O., Micali, S., and Wigderson, A.}
\newblock How to play any mental game.
\newblock In {\em Proceedings of the nineteenth annual ACM symposium on Theory
  of computing\/} (1987), ACM, pp.~218--229.

\bibitem{golub2012matrix}
{\sc Golub, G.~H., and Van~Loan, C.~F.}
\newblock {\em Matrix computations}, vol.~3.
\newblock JHU Press, 2012.

\bibitem{groth2009homomorphic}
{\sc Groth, J.}
\newblock Homomorphic trapdoor commitments to group elements.
\newblock {\em IACR Cryptology ePrint Archive 2009\/} (2009), 7.

\bibitem{Halevy09}
{\sc Halevy, A., Norvig, P., and Pereira, F.}
\newblock The unreasonable effectiveness of data.
\newblock {\em IEEE Intelligent Systems 24}, 2 (Mar. 2009), 8--12.

\bibitem{HallRidge}
{\sc Hall, R., Fienberg, S.~E., and Nardi, Y.}
\newblock Secure multiple linear regression based on homomorphic encryption.
\newblock In {\em Journal of Official Statistics\/} (2011).

\bibitem{ryoan}
{\sc Hunt, T., Zhu, Z., Xu, Y., Peter, S., and Witchel, E.}
\newblock Ryoan: A distributed sandbox for untrusted computation on secret
  data.
\newblock In {\em OSDI\/} (2016), pp.~533--549.

\bibitem{iyengartowards}
{\sc Iyengar, R., Near, J.~P., Song, D., Thakkar, O., Thakurta, A., and Wang,
  L.}
\newblock Towards practical differentially private convex optimization.
\newblock In {\em 2019 IEEE Symposium on Security and Privacy (SP)}, IEEE.

\bibitem{jagielski2018manipulating}
{\sc Jagielski, M., Oprea, A., Biggio, B., Liu, C., Nita-Rotaru, C., and Li,
  B.}
\newblock Manipulating machine learning: Poisoning attacks and countermeasures
  for regression learning.
\newblock {\em arXiv preprint arXiv:1804.00308\/} (2018).

\bibitem{Gazelle}
{\sc Juvekar, C., Vaikuntanathan, V., and Chandrakasan, A.}
\newblock Gazelle: {A} low latency framework for secure neural network
  inference.
\newblock {\em CoRR abs/1801.05507\/} (2018).

\bibitem{mascot}
{\sc Keller, M., Orsini, E., and Scholl, P.}
\newblock Mascot: faster malicious arithmetic secure computation with oblivious
  transfer.
\newblock In {\em Proceedings of the 2016 ACM SIGSAC Conference on Computer and
  Communications Security\/} (2016), ACM, pp.~830--842.

\bibitem{overdrive}
{\sc Keller, M., Pastro, V., and Rotaru, D.}
\newblock Overdrive: making spdz great again.
\newblock In {\em Annual International Conference on the Theory and
  Applications of Cryptographic Techniques\/} (2018), Springer, pp.~158--189.

\bibitem{lassoForCancer}
{\sc Kidd, A.~C., McGettrick, M., Tsim, S., Halligan, D.~L., Bylesjo, M., and
  Blyth, K.~G.}
\newblock Survival prediction in mesothelioma using a scalable lasso regression
  model: instructions for use and initial performance using clinical
  predictors.
\newblock {\em BMJ Open Respiratory Research 5}, 1 (2018).

\bibitem{spectre}
{\sc Kocher, P., Genkin, D., Gruss, D., Haas, W., Hamburg, M., Lipp, M.,
  Mangard, S., Prescher, T., Schwarz, M., and Yarom, Y.}
\newblock Spectre attacks: Exploiting speculative execution.
\newblock {\em arXiv preprint arXiv:1801.01203\/} (2018).

\bibitem{lee2017inferring}
{\sc Lee, S., Shih, M.-W., Gera, P., Kim, T., Kim, H., and Peinado, M.}
\newblock Inferring fine-grained control flow inside sgx enclaves with branch
  shadowing.
\newblock In {\em 26th USENIX Security Symposium, USENIX Security\/} (2017),
  pp.~16--18.

\bibitem{MiniONN}
{\sc Liu, J., Juuti, M., Lu, Y., and Asokan, N.}
\newblock Oblivious neural network predictions via minionn transformations.
\newblock In {\em Proceedings of the 2017 ACM SIGSAC Conference on Computer and
  Communications Security\/} (2017), ACM, pp.~619--631.

\bibitem{liu2018fine}
{\sc Liu, K., Dolan-Gavitt, B., and Garg, S.}
\newblock Fine-pruning: Defending against backdooring attacks on deep neural
  networks.
\newblock {\em arXiv preprint arXiv:1805.12185\/} (2018).

\bibitem{sgx}
{\sc McKeen, F., Alexandrovich, I., Berenzon, A., Rozas, C.~V., Shafi, H.,
  Shanbhogue, V., and Savagaonkar, U.~R.}
\newblock Innovative instructions and software model for isolated execution.
\newblock {\em HASP@ ISCA 10\/} (2013).

\bibitem{secureml}
{\sc Mohassel, P., and Zhang, Y.}
\newblock Secureml: A system for scalable privacy-preserving machine learning.
\newblock {\em IACR Cryptology ePrint Archive 2017\/} (2017), 396.

\bibitem{DJoin}
{\sc Narayan, A., and Haeberlen, A.}
\newblock Djoin: Differentially private join queries over distributed
  databases.
\newblock In {\em Proceedings of the 10th USENIX Conference on Operating
  Systems Design and Implementation\/} (2012), OSDI'12.

\bibitem{Nikolaenko}
{\sc Nikolaenko, V., Weinsberg, U., Ioannidis, S., Joye, M., Boneh, D., and
  Taft, N.}
\newblock Privacy-preserving ridge regression on hundreds of millions of
  records.
\newblock In {\em Security and Privacy (SP), 2013 IEEE Symposium on\/} (2013),
  IEEE, pp.~334--348.

\bibitem{nikolaenko2013privacy}
{\sc Nikolaenko, V., Weinsberg, U., Ioannidis, S., Joye, M., Boneh, D., and
  Taft, N.}
\newblock Privacy-preserving ridge regression on hundreds of millions of
  records.
\newblock In {\em Security and Privacy (SP), 2013 IEEE Symposium on\/} (2013),
  IEEE, pp.~334--348.

\bibitem{Paillier}
{\sc Paillier, P.}
\newblock Public-key cryptosystems based on composite degree residuosity
  classes.
\newblock In {\em EUROCRYPT\/} (1999), pp.~223--238.

\bibitem{lassoForGenomics2}
{\sc Papachristou, C., Ober, C., and Abney, M.}
\newblock A lasso penalized regression approach for genome-wide association
  analyses using related individuals: application to the genetic analysis
  workshop 19 simulated data.
\newblock {\em BMC Proceedings 10}, 7 (Oct 2016), 53.

\bibitem{chameleon}
{\sc Riazi, M.~S., Weinert, C., Tkachenko, O., Songhori, E.~M., Schneider, T.,
  and Koushanfar, F.}
\newblock Chameleon: A hybrid secure computation framework for machine learning
  applications.
\newblock Cryptology ePrint Archive, Report 2017/1164, 2017.
\newblock \url{https://eprint.iacr.org/2017/1164}.

\bibitem{sgd}
{\sc Robbins, H., and Monro, S.}
\newblock A stochastic approximation method.
\newblock In {\em Herbert Robbins Selected Papers}. Springer, 1985,
  pp.~102--109.

\bibitem{DeepSecure}
{\sc Rouhani, B.~D., Riazi, M.~S., and Koushanfar, F.}
\newblock Deepsecure: Scalable provably-secure deep learning.
\newblock {\em CoRR abs/1705.08963\/} (2017).

\bibitem{lassoForFinance}
{\sc Roy, S., Mittal, D., Basu, A., and Abraham, A.}
\newblock Stock market forecasting using lasso linear regression model, 01
  2015.

\bibitem{vc3}
{\sc Schuster, F., Costa, M., Fournet, C., Gkantsidis, C., Peinado, M.,
  Mainar-Ruiz, G., and Russinovich, M.}
\newblock Vc3: Trustworthy data analytics in the cloud using sgx.
\newblock In {\em Security and Privacy (SP), 2015 IEEE Symposium on\/} (2015),
  IEEE, pp.~38--54.

\bibitem{shmatikovmachine}
{\sc Shmatikov, V., and Song, C.}
\newblock What are machine learning models hiding?

\bibitem{ShokriShmatikov}
{\sc Shokri, R., and Shmatikov, V.}
\newblock Privacy-preserving deep learning.
\newblock In {\em CCS\/} (2015).

\bibitem{stoica2017berkeley}
{\sc Stoica, I., Song, D., Popa, R.~A., Patterson, D., Mahoney, M.~W., Katz,
  R., Joseph, A.~D., Jordan, M., Hellerstein, J.~M., Gonzalez, J.~E., et~al.}
\newblock A berkeley view of systems challenges for ai.
\newblock {\em arXiv preprint arXiv:1712.05855\/} (2017).

\bibitem{van2018foreshadow}
{\sc Van~Bulck, J., Minkin, M., Weisse, O., Genkin, D., Kasikci, B., Piessens,
  F., Silberstein, M., Wenisch, T.~F., Yarom, Y., and Strackx, R.}
\newblock Foreshadow: Extracting the keys to the intel sgx kingdom with
  transient out-of-order execution.
\newblock In {\em Proceedings of the 27th USENIX Security Symposium. USENIX
  Association\/} (2018).

\bibitem{wang2017global}
{\sc Wang, X., Ranellucci, S., and Katz, J.}
\newblock Global-scale secure multiparty computation.
\newblock In {\em Proceedings of the 2017 ACM SIGSAC Conference on Computer and
  Communications Security\/} (2017), ACM, pp.~39--56.

\bibitem{yao1982protocols}
{\sc Yao, A.~C.}
\newblock Protocols for secure computations.
\newblock In {\em Foundations of Computer Science, 1982. SFCS'08. 23rd Annual
  Symposium on\/} (1982), IEEE, pp.~160--164.

\bibitem{opaque}
{\sc Zheng, W., Dave, A., Beekman, J.~G., Popa, R.~A., Gonzalez, J.~E., and
  Stoica, I.}
\newblock Opaque: An oblivious and encrypted distributed analytics platform.
\newblock In {\em USENIX Symposium of Networked Systems Design and
  Implementation (NDSI)\/} (2017), pp.~283--298.

\end{thebibliography}

\iffull
\appendix

\section{ADMM derivations}
\label{sec:admm:derivation}
Ridge regression solves a similar problem as LASSO, except with L2 regularization.
Given dataset $(X, \vec{y})$ where $X$ is the feature matrix and $\vec{y}$ is the prediction vector, ridge regression optimizes 
$\argmin_{\vec{w}} \dfrac{1}{2} \|X \vec{w} - \vec{y}\|_{2}^{2} + \lambda \|\vec{w}\|_{2}$.
Splitting the weights into $\vec{w}$ and $\vec{z}$, we have 
\begin{align*}
&\text{minimize } \dfrac{1}{2} \| X \vec{w} - \vec{y} \|_{2}^{2} + \lambda \| \vec{z} \|_{2} \\
&\text{subject to } \vec{w} - \vec{z} = 0
\end{align*}

We first find the augmented Lagrangian 
\begin{align*}
\mathscr{L}(\vec{w}, \vec{z}, \vec{v}) &= \dfrac{1}{2} \| X \vec{w} - \vec{y}\|_{2}^{2} + \lambda \| \vec{z} \|_{2}\\
&+ \vec{v}^{T}(\vec{w} - \vec{z}) + \dfrac{\rho}{2} \| \vec{w} - \vec{z}\|_{2}^{2}
\end{align*}
where $\vec{w}$ and $\vec{z}$ are the primal weight vectors, and $\vec{v}$ is the dual weight vector.
To simply the equations, we replace $\vec{v}$ with the scaled dual variable $\vec{u}$ where $\vec{u} = (1/\rho) \vec{v}$.
The update equations come out to 
\begin{align*}
\vec{w}^{k+1} &= \argmin_{\vec{w}} \dfrac{1}{2}\| X \vec{w} - \vec{y}\|_2^2 
+ (\rho / 2) \| \vec{w} - \vec{z}^{k} + \vec{u}^{k}\|_2^2)\\
\vec{z}^{k+1} &= \argmin_{\vec{z}} \lambda \| \vec{z} \|_{2}^{2} + (\rho/2) (\vec{w}^{k+1} + \vec{z} + \vec{u}^{k})\\
\vec{u}^{k+1} &= \vec{u}^{k} + \vec{w}^{k+1} + \vec{z}^{k+1}
\end{align*}

Since our loss function is decomposable based on data blocks, we can apply the generic global variable consensus ADMM algorithm and find

\begin{align*}
\vec{w}_{i}^{k+1} &= \argmin_{\vec{w}_{i}} \dfrac{1}{2}\| X_i \vec{w} - \vec{y}\|_2^2 
+ (\rho / 2) \| \vec{w}_{i} - \vec{z}^{k} + \vec{u}^{k}\|_2^2)\\
\vec{z}^{k+1} &= \argmin_{\vec{z}} \lambda \| \vec{z} \|_{2}^{2} + (m \rho/2) \|\vec{z} - \vec{\bar{w}}^{k+1} - \vec{\bar{u}}^{k}\|_{2}^{2}\\
\vec{u}^{k+1} &= \vec{u}^{k} + \vec{w}^{k+1} - \vec{z}^{k+1}
\end{align*}

Thus, the $\vec{w}$ update is
\begin{align*}
\vec{w}_{i}^{k+1} &= (X_i^{T} X_i + \rho I)^{-1}(X_i^{T}\vec{y}_i + \rho(\vec{z}^{k} - \vec{u}_i^{k})) \\
&+ (\rho/2) \| \vec{w}_i - \vec{z}^{k} + \vec{u}_i^{k} \|_{2}^{2}\\
\vec{z}^{k+1} &= \dfrac{\rho}{2\lambda/m + \rho} (\vec{\bar{w}}^{k+1} + \vec{\bar{u}}^{k})\\
\vec{u}_i^{k+1} &= \vec{u}_i^{k} + \vec{x}_i^{k+1} - \vec{z}^{k+1}
\end{align*}

Therefore, the $\vec{w}_i$ update is the same as LASSO and can be computed using the same techniques.
The $\vec{z}$ update is actually linear and does not require comparisons, though MPC is still required for reducing the scaling factors accumulated during fixed point integer multiplications.


\section{Security proofs}
\subsection{Definitions}
We first define the MPC model. 
The full definitions are taken from~\cite{cramer2001multiparty,canetti2006security}, so please refer to those papers for more details.

\noindent\textbf{Real world model.} Let $\pi$ be an $n$-party protocol.
The protocol is executed on an open broadcast network with static, active, and rushing adversary $\adv$ (rushing means that the adversaries can send their messages after seeing all of the honest parties' messages).
The number of adversaries can be a majority of the participants.
Let $\kappa$ be the security parameter.
Each party $P_i$ has public input $x_i^{p}$ and secret input $x_i^{s}$, as well as public output $y_i^{p}$ and secret output $y_i^{s}$.
The adversary $\adv$ receives the public input and output of all parties.

Let $\vec{x} = (x_1^{s}, x_1^{p}, \dots, x_n^{s}, x_n^{p})$ be the parties' input, 
and let $\vec{r} = (r_1, \dots, r_n, r_{\adv{}})$ be the parties' and the adversary's private input randomness. 
Let $C \subset M$ be the corrupted parties, and let $a \in \{0, 1\}^{*}$ be the adversary's auxiliary input.
Let $H \subset M$ be the honest parties.
Therefore, we have that $H + C = M$.

By $\text{ADVR}_{\pi, \adv}(\kappa, \vec{x}, C, a, \vec{r})$ and $\text{EXEC}_{\pi, \adv}(\kappa, \vec{x}, C, a, \vec{r})_i$ we denote the output of the adversary $\adv$ and the output of party $P_i$, respectively, after a real world execution of $\pi$ with the given input under attack from $\adv$.
%
Let
\begin{align}
\text{EXEC}_{\pi, \adv}(\kappa, \vec{x}, C, a, \vec{r}) = &(\text{ADVR}_{\pi, \adv}(k, \vec{x}, C, a, \vec{r}), \\
& \text{EXEC}_{\pi, \adv}(\kappa, \vec{x}, C, a, \vec{r})_1,\\
&\dots,\\
& \text{EXEC}_{\pi, \adv}(\kappa, \vec{x}, C, a, \vec{r})_n)
\end{align}
This is simply the union of the different parties' and the adversary's real world output distribution.
Denote by $\text{EXEC}_{\pi, \adv}(\kappa, \vec{x}, C, a)$ the random variable $\text{EXEC}_{\pi, \adv}(k, \vec{x}, C, a, \vec{r})$, where $\vec{r}$ is chosen uniformly random.
We define the distribution ensemble with security parameter $\kappa$ and index $(\vec{x}, C, a)$ by 
\begin{align}
\text{EXEC}_{\pi, \adv} = \{ \text{EXEC}_{\pi, \adv}(\kappa, \vec{x}, C, a, \textbf{r})_i\}_{\kappa \in M, \vec{x} \in (\{0, 1\}^{*})^{2n}, a \in \{0, 1\}^{*}, i \in H}
\end{align}

\noindent\textbf{Ideal world model.}
Let $f: M \times (\{0, 1\}^{*})^{2n} \times \{0, 1\}^{*} \rightarrow (\{0, 1\}^{*})^{2n}$ be a probabilistic $n$-party function computable in probabilistic polynomial time (PPT).
The inputs and outputs are defined as $(y_1^{s}, y_1^{p}, \dots, y_n^{s}, y_n^{p}) \leftarrow f(\kappa, x_1^{s}, x_1^{p}, \dots, x_n^{s}, x_n^{p}, r)$, where $\kappa$ is the security parameter and $r$ is the random input.
In the ideal world, the parties send their inputs to a trusted third party $T$ that chooses a uniformly random $r$, computes $f$ on these inputs and returns $(y_i^{s}, y_i^{p})$ to $P_i$.

The active static ideal world adversary $\adv_{I}$ sees all $x_i^{p}$ values, as well as $x_i^{s}$ for all corrupted parties.
The adversary then substitutes the values $(x_i^{s}, x_i^{p})$ for the corrupted parties by values of his choice $(x_i^{s'}, x_i^{p'})$.
We set $(x_i^{s'}, x_i^{p'}) = (x_i^{s}, x_i^{p})$ for the honest parties.
The ideal function $f$ is evaluated on $(\kappa, x_1^{s'}, x_1^{p'}, \dots, x_n^{s'}, x_n^{p'}, r)$ via an oracle call.
Each party receives output $(y_i^{s}, y_i^{p})$, and the adversary sees $y_i^{p}$ for all parties as well as $y_i^{s}$ for all corrupted parties.

Similar to the real world execution, we define $\text{ADVR}_{\pi, \adv_i}(\kappa, \vec{x}, C, a, \vec{r})$ and $\text{IDEAL}_{\pi, \adv_i}(\kappa, \vec{x}, C, a, \vec{r})_i$ we denote the output of the adversary $\adv_i$ and the output of party $P_i$, respectively, after an ideal world execution with the given input under attack from $\adv_i$.
Let
\begin{align}
\text{IDEAL}_{f, \adv_{I}}(\kappa, \vec{x}, C, a, \vec{r}) = (&\text{ADVR}_{f, \adv_{I}}(\kappa, \vec{x}, C, a, \vec{r}), \\
&\text{IDEAL}_{f, \adv_{I}}(\kappa, \vec{x}, C, a, \vec{r})_1,\\
&\text{IDEAL}_{f, \adv_{I}}(\kappa, \vec{x}, C, a, \vec{r})_2,\\
&\dots,\\
&\text{IDEAL}_{f, \adv_{I}}(\kappa, \vec{x}, C, a, \vec{r})_n)
\end{align}
denote the collective output distribution of the parties and the adversary.
Define a distribution ensemble by 
\begin{align}
\text{IDEAL}_{f, \adv_I} = \{ \text{IDEAL}_{f, \adv_I}(\kappa, \vec{x}, C, a)_i \}_{\kappa \in N, \vec{x} \in (\{0, 1\}^{*})^{2n}, a \in \{0, 1\}^{*}, i \in H}
\end{align}

\noindent\textbf{Hybrid model.}
In the $(g_1, \dots, g_l)$-hybrid model, the execution of a protocol $\pi$ proceeds in the real-life model, except that the parties have access to a trusted party $T$ for evaluating the $n$-party functions $g_1, \dots, g_l$.
These ideal evaluations proceed as in the ideal world model.
The distribution ensemble is 
\begin{align}
\text{EXEC}_{\pi, \adv}^{g_1, \dots, g_l} = \{ \text{EXEC}_{\pi, \adv}^{g_1, \dots, g_l}(\kappa, \vec{x}, C, a)_i \}_{\kappa \in N, \vec{x} \in (\{0, 1\}^{*})^{2n}, a \in \{0, 1\}^{*}, i \in H}
\end{align}

Security can be defined by requiring a real world execution or a $(g_1, \dots, g_l)$-hybrid execution of a protocol $\pi$ for computing an ideal functionality $f$ to reveal no more information to an adversary than does an ideal execution of $f$.
We can define the real world model by the $()$-hybrid model.

\begin{mydef}
Let $f$ be an $n$-party function, let $\pi$ be an $n$-party protocol
We say that $\pi$ securely evaluates $f$ in the $(g_1, \dots, g_l)$-hybrid model if for any active static $(g_1, \dots, g_l)$-hybrid adversary $\adv$, which corrupts only subsets of $C$, there exists a static active ideal-model adversary $\mathscr{S}$ such that $\text{IDEAL}_{f, \mathscr{S}} \approx_{c} \text{EXEC}_{\pi, \adv}^{g_1, \dots, g_l}$.
\end{mydef}

Finally, we utilize the modular composition operation that was defined in~\cite{canetti2006security}.
The modular composition theorem (informally) states that if $\pi$ $\Gamma$-securely evaluates $f$ in the $(g_1, \dots, g_l)$-hybrid model and $\pi_{g_i}$ $\Gamma$-securely evaluates $g_i$ in the $(g_1, \dots, g_{i-1}, g_{i+1}, \dots, g_l)$-hybrid model, 
then the protocol $\pi^{'}$, which follows protocol $\pi$ except with oracle calls to $g_i$ replaced by executions of the protocol $\pi_{g_i}$, $\Gamma$-securely evaluates $f$ in the $(g_1, \dots, g_{i-1}, g_{i+1}, \dots, g_l)$-hybrid model.

Next, we describe some essential lemmas and existing protocols that we use.

\begin{theorem}[Schwartz-Zippel]
\label{th:schwartz:zippel}
Let $P \in F[x_1, x_2, \dots, x_n]$ be a non-zero polynomial of total degree $d > 0$ over a field $F$.
Let $S$ be a finite subset of $F$ and let $r_1, r_2, \dots, r_n$ be selected at random independently and uniformly from $S$.
Then $\Pr[P(r_1, r_2, \dots, r_n) = 0] \leq \dfrac{d}{|S|}$.
\end{theorem}

\begin{lemma}[Smudging lemma]
\label{lemma:smudging}
Let $B_1 = B_1(\kappa)$, and $B_2 = B_2(\kappa)$ be positive integers and let $\epsilon_1 \in [-B_1, B_1]$ be a fixed integer.
Let $e_2 \in_{R} [-B_2, B_2]$ be chosen uniformly at random. 
Then, the distribution of $e_2$ is statistically indistinguishable from that of $e_2 + e_1$ as long as $B_1 / B_2 = \text{neg}(\kappa)$.
\end{lemma}

\cref{lemma:smudging} is used for arguing statistical indistinguishability between two distributions.

Next, we list three existing zero-knowledge proofs that serve as building blocks in our system.
They are all $\Sigma$ protocols~\cite{damgaard2002sigma}, which assume that the verifier is honest.
However, they can be transformed into full zero knowledge, as we explain in detail later.
Since this is taken from existing literature, we will not re-derive the simulators here, and instead assume access to simulators for all three protocols.

\begin{protocol}[Paillier proof of plaintext knowledge]
\label{prot:paillier:pok}
A protocol for proving plaintext knowledge of a Paillier ciphertext $\enc{a}$~\cite{cramer2001multiparty}.
\end{protocol}

\begin{protocol}[Paillier multiplication proof] 
\label{prot:paillier:mult}
A protocol for: given $\enc{\alpha}, \enc{a}, \enc{b}$, prove that $\enc{b}$ indeed encrypts $\alpha \cdot a$ and that the prover has plaintext knowledge of $\enc{\alpha}$~\cite{cramer2001multiparty}.
\end{protocol}

\begin{protocol}[Encryption interval proof]
\label{prot:interval:proof}
An efficient interval proof that proves an encrypted value lies within an interval.
\end{protocol}

Note that in order to construct~\cref{prot:interval:proof}, we combine two existing protocols together.
The first protocol is an interval proof for a committed value~\cite{boudot2000efficient}.
The second protocol is an additional proof that proves the equality of plaintexts under a commitment and an encryption~\cite{damgaard2002client}.

To turn our honest verifier proofs into full zero knowledge (as well as non-malleable and concurrent), we utilize an existing transformation~\cite{garay2003strengthening}.
%
%
Informally, the transformation does two things.
First, the random challenge in a $\Sigma$ protocol is first generated by the verifier giving a challenge, then the prover proves an OR protocol given this challenge.
The OR protocol consists of the actual statement to be proved, plus a signature proof.
A simulator can simply simulate the second protocol in the OR protocol, instead of the main proof.
Second, the witness for the proof is extracted via encrypting it under a public key generated from the common reference string.
During the simulation, the simulator is able to generate the encryption key parameters, and can thus decrypt the encryption and extract an adversary's input.
This means no rewinding is needed either for simulation or extraction.
We denote the parameter generation functionality as $f_{\text{crs}}$.
In \sys's design, we transform a $\Sigma$ protocol $S$ using this method, then use the simulator $S_{zk}(S)$ and extractor $E_{zk}(S)$ in our proofs for simulation and extraction of the adversary's secrets.

Finally, we assume that we have access to a threshold encryption scheme that can provably decrypt a ciphertext for our threshold structure.
The scheme is described in~\cite{fouque2000sharing}, and we do not provide further proofs for this protocol.
In our MPC simulation, we assume that we have a simulator $S_{dec}$ for the decryption protocol.

\subsection{Proofs}
\label{sec:proofs}

\begin{theorem}
\label{proof:matrix:pok}
Protocol described in~\cref{prot:paillier:matmul:proof} is an honest verifier zero knowledge proof of plaintext knowledge for a committed matrix $\enc{\matr{X}}$.
\end{theorem}

\begin{proof}[Proof sketch]
The Paillier ciphertext proof of plaintext knowledge is a $\Sigma$ protocol.
Correctness, soundness, and simulation arguments are given in~\cite{cramer2001multiparty}.
Using $S_{zk}$ and $E_{zk}$, we can achieve full zero knowledge in the concurrent setting by simply proving knowledge for each element in the matrix using the ciphertext proof of plaintext knowledge.
\end{proof}

\begin{theorem}
\label{proof:paillier:matmul:proof}
\cref{gadget:knowsone}, with applied transformation from~\cite{garay2003strengthening}, is a zero knowledge argument for proving the following: given a public committed $\enc{\matr{X}}$, an encrypted $\enc{\matr{Y}}$, and $\enc{\matr{Z}}$ prove that the prover knows $\matr{X}$ and that $Z = XY$ under standard cryptographic assumptions.
\end{theorem}

\begin{proof}[Proof sketch]
To prove this theorem, we first prove the security of the honest verifier version of this protocol.
The argument itself contains several parts.
The first is a proof of plaintext knowledge of the matrix $\matr{X}$, which is applied straightforwardly from~\cref{proof:matrix:pok}.
This allows us to extract the content of the commitment.
The second is a matrix multiplication proof, which consists of a reduction and ciphertext multiplication proofs.

Completeness is straightforward to see since we simply follow the computation of matrix multiplication, except that $\matr{Y}$ is encrypted.
If $\enc{\matr{Z}} = \matr{X} \enc{\matr{Y}}$, then $\vec{t} \enc{\matr{Z}} = \vec{t} \matr{X} \enc{\matr{Y}}$.
Soundness can be proved in two steps.
The first step utilizes the Schwartz-Zippel lemma~\cref{th:schwartz:zippel}.
Given a random vector $\vec{t} \in [0, 2^{l}]$ where $l = |n| - 2|p| - \log{M}$, we will verify that $\vec{t} \matr{Z} = \vec{t} \matr{X} \matr{Y}$.
This step reduces the problem to verifying matrix-vector multiplication instead of matrix-matrix multiplication.
Using the lemma, we can view this as a multivariable polynomial equality testing problem.
Therefore, an inequality will correctly pass with probability $d / |S|$.
Since the commitments to $\matr{X}$, $\matr{Y}$, and $\matr{Z}$ are homomorphic, the prover and the verifier can calculate $\enc{\vec{t} \matr{X}}$ and $\enc{\vec{t} \matr{Z}}$ independently.
The second step of the soundness argument comes from the fact that after this transformation, we execute individual ciphertext multiplication proofs, which are $\Sigma$ protocols themselves.
These $\Sigma$-protocols are used for the individual products of a dot product (these are proved independently), as well as for proving the summation of these individual products.
The summation itself can be computed by the verifier directly via the homomorphic properties of the ciphertexts.
Then the summation proof is another $\Sigma$-protocol.
$\Sigma$-protocols satisfy the special soundness property~\cite{damgaard2002sigma}, which means that a cheating verifier can cheat with probability $2^{-t}$, where $t$ indicates the length of the challenge.
Therefore, the probability of cheating is overall negligible via union bound.
To simulate the matrix multiplication proofs, we first use the randomness on the verifier input tape to construct $\vec{t}$.
Assuming that the simulator for~\cref{prot:paillier:mult} is $S_{mult}$, we then use the challenges for the ciphertext multiplication protocols and feed each into the simulator $S_{mult}$ that simulates the $\Sigma$-protocol for ciphertext multiplication.
%
Finally, to make this entire argument full zero knowledge, we apply the protocol transformation from~\cite{garay2003strengthening}.
\end{proof}


\begin{theorem}
\label{proof:paillier:plaintext:matmul:proof}
\cref{gadget:knowall}, with applied transformation from~\cite{garay2003strengthening}, is a zero knowledge argument for proving the following: given committed matrices $\enc{\matr{X}}$, $\enc{\matr{Y}}$, and $\enc{\matr{Z}}$, prove that $\matr{X} \matr{Y} = \matr{Z}$ and that the prover knows the committed values $X$ and $Y$ under standard cryptographic assumptions.
\end{theorem}

\begin{proof}[Proof sketch]
The proof is a straightforward combination of \cref{proof:matrix:pok} and \cref{proof:paillier:matmul:proof}.
\end{proof}

\begin{theorem}\label{proof:additive:sharing}
\cref{gadget:additive:sharing} has a simulator $S_{a}$ such that $S_a$'s distribution is statistically indistinguishable from \cref{gadget:additive:sharing}'s real world execution.
\end{theorem}
\begin{proof}[Proof sketch]
First, we must construct such a simulator $S_a$.
To do so, we modify a similar simulator from~\cite{cramer2001multiparty}.
Let $M$ denote all of the parties in the protocol.
The simulator runs the following:
\begin{enumerate}
\item\label{simulator:ass:step1} Let $s$ be the smallest index of an honest party and let $H^{'}$ be the set of the remaining honest parties.
For each honest party in $H^{'}$, generate $r_i$ and $\enc{r_i}$ correctly.
For party $s$, choose $r_s^{'}$ uniformly random from $[0, 2^{|p| + \kappa}]$, and let $\enc{r_s} = \enc{r_s^{'} - a}$.
\item Hand the values $\enc{r_i}_{i \in H}$ to the adversary and receive from the adversary $\enc{r_i}$.
\item Run the augmented proofs of knowledge ($S_{zk}(P_{pok})$) from the adversaries and simulate the proof for party $s$ (since the simulator does not know the plaintext value of $r_s$).
If any proof fails from the adversary, abort.
Otherwise, continue and use the augmented extractor ($E_{zk}(P_{pok})$)to extract the adversary's inputs $r_i$.
\item Compute $e = \sum_{i \neq s} (r_i) + r_s^{'} = \sum_i r_i + a$. 
\item Simulate a call to decrypt using $S_{dec}$. 
Note that the decrypted value is exactly $e$ due to the relation described in the previous step.
\item Simulator computes the shares as indicated in the original protocol except for party $s$, which sets its value to $a_s^{'} = (a_s - a) \mod p$.
\end{enumerate}

We now prove that the simulator's distribution is statistically indistinguishable from the real world execution's distribution.

Note that other than $\enc{r_s}$, the rest of $S_a$'s simulation follows exactly from the real world execution $A$, and thus is distributed exactly the same as the execution.

In simulation step~\ref{simulator:ass:step1}, $\enc{r_s}$ encrypts $r_s$.
Given a plaintext that is within $[0, p)$, 
using~\cref{lemma:smudging} we know that $r_s$ and $r_s^{'}$ are statistically indistinguishable from each other.
This means that both $r_s$'s and $r_s^{'}$'s distributions are statistically close to being uniformly drawn from the interval $[0, 2^{|p| + \kappa}]$.
Since we always blind encryptions, this means that $\enc{r_s}$ is a random encryption of a statistically indistinguishable uniformly random element from $[0, 2^{|p| + \kappa}]$.
Therefore, $\enc{d_s}$'s distribution in the simulator is statistically indistinguishable from the corresponding distribution in the execution.

The distributions of the real world proofs and the simulated proofs follow straightforwardly from~\cite{garay2003strengthening}.

Finally, we know that $a_s^{'} = a_s - a \mod p$.
The real world execution's share is $a_s$.
Since $r_s$ and $r_s^{'}$ are statistically indistinguishable from the uniformly random distribution  and $a < p$, the $r_s$'s and $r_s^{'}$'s distributions after applying modulo $p$ are also statistically indistinguishable.
Therefore, the $a_s^{'}$ and $a_s$ distributions are also statistically indistinguishable.

\end{proof}

Next, we first define $f_{crs}$ and $f_{SPDZ}$, two ideal functionalities.

\begin{enumerate}
\item $f_{\text{crs}}$: an ideal functionality that generates a common reference string, as well as secret inputs to the parties. 
As mentioned before, this functionality is used for the augmented proofs so that the extractor can extract the adversary's inputs by simple decryption.
\item $f_{\text{SPDZ}}$: an ideal functionality that computes the ADMM consensus phase using SPDZ.
\end{enumerate}

\begin{theorem}
$f_{\text{ADMM}}$ in the $(f_{\text{crs}}, f_{\text{SPDZ}})$-hybrid model under standard cryptographic assumptions, against a malicious adversary who can statically corrupt up to $m-1$ out of $m$ parties.
\end{theorem}

\begin{proof}[Proof sketch]
To prove \sys's security, we first start the proof by constructing a simulator for \sys's two phases: input preparation and model compute.
Next, we prove that the simulator's distribution is indistinguishable from the real world execution's distribution.
Thus, we prove security in the ($f_{\text{crs}}, f_{\text{SPDZ}}$)-hybrid model.

First, we construct a simulator $S$ that first simulates the input preparation phase, followed by the model compute phase.

\begin{enumerate}
\item $S$ simulates $f_{\text{crs}}$ by generating the public key $\pk_{\text{crs}}$ and a corresponding secret key $\sk_{\text{crs}}$. 
These parameters are used for the interactive proof transformations so that we are able to extract the secrets from the proofs of knowledge.
\item $S$ next generates the threshold encryption parameters.
The public key $\pk$ is handed to every party.
The secret key shares $\skshare{i}$ are handed to each party as well.
Discard $\skshare{i}$ for the honest parties.
%
%
\item Next, $S$ starts simulating the input preparation phase.
It receives matrix inputs, as well as interval proofs of knowledge from the adversary $\mathscr{A}$.
It also generates dummy inputs for the honest parties, e.g., encrypting vectors and matrices of $0$.
\item If the proofs of knowledge from $\mathscr{A}$ pass, then $S$ extracts the inputs using the augmented extractors from~\cref{proof:matrix:pok}.
Otherwise, abort.

\item $S$ hands the inputs from the adversary to the ideal functionality $f_{\text{ADMM}}$, which will output the final weights $\vec{w_{\text{final}}}$ to the simulator.

\item The first two steps in the input preparation phase utilize matrix plaintext multiplication proofs.
For each honest party, $S$ simply proves using its dummy inputs and simulates the appropriate proofs.
$S$ should also receive proofs from the adversary.
If the proofs pass, continue.
Otherwise, abort the simulation.
\item The simulator continues to step 3 of the input preparation phase.
This step has two different proofs.
The first proof is a matrix multiplication proof between $\matr{V^{T}}$ and $\matr{V}$.
This can be simulated like the previous steps.
The next step is an element-wise interval proof with respect to the identity matrix.
The simulator can again simulate this step using the simulator for the interval proof. 
Next, $S$ verifies the proofs from the malicious parties.
If the proofs pass, move on.
Otherwise, abort the simulation.
\item Something similar can be done to prove and verify the step 4 of input preparation phase from each party, since the proofs utilized are similar to those used in step 3.
%
%
\item $S$ now begins simulating the model compute phase.
\item Initialize the encrypted weights ($\vec{w}, \vec{z}, \vec{u}$) to be the zero vector.
\item $\text{for } i \text{ in } admm\_iters$:
\begin{enumerate}
\item The first step in the iteration is the local compute phase.
$S$ simulates the honest parties by correctly executing the matrix multiplications using the dummy input matrices and the encrypted weight and produces both the encrypted results and the multiplication proofs from~\cref{prot:paillier:mult}.
$S$ also receives a set of encrypted results and multiplication proofs from the adversary for the corrupted parties.
$S$ can verify that $\mathscr{A}$ has indeed executed the matrix multiplication proofs correctly.
If any of the proofs doesn't pass, the simulator aborts.
If not all of the ciphertexts are distinct, the simulator also aborts.
\item 
Now, $S$ needs to simulate the additive sharing protocol for each party.
This can be done by invoking the simulator from~\cref{proof:additive:sharing}.
\item After the secret sharing process is complete, all parties need to publish encryptions of their secret shares.
$S$ publishes encryptions of these shares for the honest parties and the appropriate interval proofs of knowledge.
The malicious parties also publish encrypted shares and their interval proofs of knowledge.
If the interval proofs of knowledge do not pass for the adversary, then abort the computation.
Otherwise, $S$ extracts all of the encrypted shares from the adversary.
Here, $S$ also calculates the \emph{expected} shares from the adversary.
This can be calculated because $S$ knows the randomness $r_i$'s used by $\mathscr{A}$.
\item $S$ needs to now simulate a call to the $f_{SPDZ}$ oracle.
The simulator first picks a random $\alpha \in \mathds{Z}_{p}$, which serves as the global MAC key.
Then it splits $\alpha$ into random shares and gives one share to each party.
If $S$ is generating shares for the input values, it will simply generate MACs $\gamma(a)_i$ for the shares it receives from running $S_a$.
Otherwise, $S$ then generates random SPDZ shares and MAC shares $\gamma(a)_i$.
In both cases, the values shares and the MAC shares satisfy the SPDZ invariant: $\alpha (\sum_i a_i) - (\sum_i \gamma(a)_i)$.
\item Each party publishes encryptions of all SPDZ input and output shares, as well as their MAC shares.
The simulator $S$ will simulate the honest parties' output by releasing those encryptions and interval proofs of knowledge.
$S$ also receives the appropriate encryptions and interval proofs of knowledge from the adversary.
Run the extractor to extract the contents of the adversary.
Again, keep track of the shares distributed to the adversary, as well as the extracted shares that were committed by the adversary.
\item If the iteration is the last iteration, then $S$ simulates a call to SPDZ by splitting the $\vec{w}_{\text{final}}$ shares into random shares, as well as creating the corresponding MAC shares to satisfy the relation with the global key $\alpha$.
\end{enumerate}
\item \label{simulator:mpc:checks}
After the iterations, $S$ needs to simulate the MPC checks described in~\cref{sec:model:release}, where we need to prove modular equality for a set of encrypted values and shares.
$S$ runs the following in parallel, for each equality equation that needs to be proven:
\begin{itemize}
\item In the original protocol, we have $\enc{a}, \enc{b_i}, \enc{c_i}$ where we want to prove that $a \equiv \sum_i b_i \mod p$ and $\alpha (\sum_i b_i) \equiv \sum_i c_i \mod p$.
First, $S$ simulates the proof of the first equality.
$S$ follows the protocol by verifying the interval proofs of knowledge from every party.
If a malicious party's proof fails, then abort.

\item $S$ computes the cipher $\enc{a - \sum_i b_i}$ by directly operating on the known ciphertexts.

\item $S$ follows the protocol described in~\cref{sec:model:release} by following the protocol exactly.
$S$ picks random $s_i$'s and generates interval proofs of knowledge.
$S$ also receives $\enc{s_i}$ from the adversary and the corresponding interval proofs of knowledge.
Extract the $\enc{s_i}$ values from the adversary if the proofs are verified.

\item
Before $S$ releases the decrypted value, it needs to know what value to release.
To do so, $S$ needs to compare an encrypted value $a$ and its secret shares $b_i$ (each party $P_i$ retains $b_i$).
We want to make sure that $a \equiv \sum_i b_i \mod p$.
While the simulator does not know $a$, it does know some information when $a$ was first split into shares.
More specifically, $S$ knows the adversary's generated randomness during the additive secret sharing, the decrypted number $e$, as well as the shares committed by the adversary.
The equation $a + \sum_i r_i \equiv e \mod p$ holds because everyone multiplies the published $\enc{r_i}$ with $\enc{a}$ to get $\enc{a + \sum_i r_i}$, and the decryption process is simulated and verified.
This means that $a \equiv e - \sum_i r_i \mod p$.
Let's assume that an adversary alters one of its input shares $b_j$ to $b_j^{'}$ for parties in set $\mathscr{A}$.
Then the difference between $a$ and the $b_i$ shares is simply 
\begin{align*}
a - \sum_{i \not\subset \mathscr{A}} b_i - \sum_{i \subset \mathscr{A}} b_i^{'}  \mod p & \equiv e - \sum_i r_i - \sum_{i \not\subset \mathscr{A}} b_i - \sum_{i \subset \mathscr{A}} b_i^{'} \mod p\\
& \equiv (e - r_0) + \sum_{i \neq 0} (- r_i) - \sum_{i \not\subset \mathscr{A}} b_i - \sum_{i \subset \mathscr{A}} b_i^{'} \mod p\\
& \equiv b_0 + \sum_{i \neq 0}  b_i - \sum_{i \not\subset \mathscr{A}} b_i - \sum_{i \subset \mathscr{A}} b_i^{'} \mod p\\
& \equiv \sum_i b_i - \sum_{i \not\subset \mathscr{A}} b_i - \sum_{i \subset \mathscr{A}} b_i^{'} \mod p\\
& \equiv \sum_{i \subset \mathscr{A}} (b_i - b_i^{'}) \mod p
\end{align*}

Hence, the difference modulo $p$ is simply the the difference in the changes in the adversary's shares, and is completely independent from the honest parties' values.
Let this value be $v$.

\item The next step is simple: $S$ simply simulates $S_{dec}$ and releases the value $v + \sum_i (s_i p)$.

\item $S$ can follow a similar protocol for checking the SPDZ shares $b_i$ and the MACs $c_i$.

\end{itemize}
\item Finally, if the previous step executes successfully, then the parties will release their plaintext shares of $\vec{w}_{\text{final}}$ by decommitting to the encrypted ciphers of those shares and publishing their plaintext shares.
\end{enumerate}

We now prove that the distribution of the simulator is statistically indistinguishable from the distribution of the real world execution.
To do so, we construct hybrid distributions.

\noindent\textbf{Hybrid 1} This is the real world execution.

\noindent\textbf{Hybrid 2} Same as hybrid 1, except replace the proofs from the input preparation phase with simulators.

Hybrid 1 and 2 are indistinguishable because of the properties of the zero-knowledge proofs we utilize (see~\cref{proof:matrix:pok} and~\cref{proof:paillier:matmul:proof}).

\noindent\textbf{Hybrid 3} Same as the previous hybrid, except the rest of the proofs are run with the simulated proofs instead, and the secret sharing is replaced by the simulator $S_a$.
However, step~\ref{simulator:mpc:checks} is still run with the real world execution.

Hybrid 2 and 3 are statistically indistinguishable because of the properties of the zero-knowledge proofs (\cref{proof:matrix:pok}, \cref{proof:paillier:matmul:proof}), and the fact that the secret sharing is also simulatable (\cref{proof:additive:sharing}).

\noindent\textbf{Hybrid 4} Same as the previous hybrid, except swap out the real world execution with step~\ref{simulator:mpc:checks} described by the simulator.

Hybrid 3 and 4 are statistically indistinguishable.
The abort probabilities in step~\ref{simulator:mpc:checks} are based on the release values, so we just need to argue that the release value distributions are statistically indistinguishable.

Since the simulator is always able to extract the adversary's values, it will always be able to calculate the correct answer in step~\ref{simulator:mpc:checks}.
The real world execution, on the other hand, could potentially have a different answer if any of the zero-knowledge proofs fails to correctly detect wrong behavior.
Therefore, if the zero-knowledge proofs are working correctly, then:
\begin{enumerate}
\item If the execution does not abort, then the decrypted values must be both divisible by $p$.
The two values are statistically indistinguishable because of~\cref{lemma:smudging}.
\item If the execution does abort, then $S$'s output value would be $v$.
If the proofs pass, then from the argument made in step~\ref{simulator:mpc:checks} where we know that the decrypted value modulo $p$ is exactly the same as $v$. 
Furthermore, if we subtract $v$ from the two values, the new values have the same distribution by the argument in the prior step.
\end{enumerate}

If any proof does fail to detect a malicious adversary cheating, then the released answer could be potentially different.
However, this will happen with negligible probability because of the properties of the zero-knowledge proofs we are using.
Therefore, the hybrids are statistically indistinguishable.

\noindent\textbf{Hybrid 5} 
First, define $\boxplus$ to be ciphertext addition, $\boxminus$ to be ciphertext subtraction, and $\boxdot$ to be ciphertext multiplication.
Replace the input encryptions of the honest parties with encryptions where $\enc{x_i}$ is transformed into $\enc{0}$ is transformed into $\text{Blind}(\enc{0} \boxdot \enc{b}) \boxplus (\enc{x_i}\boxdot (\enc{1} \boxminus \enc{b}))$, where $b = 0$.
Hybrid 4 and 5 are computationally indistinguishable because the inputs used by the honest parties have not changed from the previous hybrid.
With the guarantees of the encryption algorithm, the encrypted ciphers from the two hybrids are also indistinguishable.

\noindent\textbf{Hybrid 6} Replace the input encryptions of the honest parties with encryptions of $0$ (so same as the simulator), except with additional randomizers such that the encryption of an input $\enc{0}$ is transformed into $\text{Blind}(\enc{0} \boxdot \enc{b}) \boxplus (\enc{x_i}\boxdot (\enc{1} \boxminus \enc{b}))$, where $b = 1$.

Hybrid 5 and 6 are indistinguishable because one could use a distinguisher $D$ to break the underlying encryption scheme. 
Since the encrypted ciphertexts are randomized and only differ by the value of $b$ (whether it is $0$ or $1$), if one were to build such a distinguisher $D$,  then $D$ can also distinguish whether $b = 0$ or $b = 1$.
This breaks the semantic encryption scheme.

\noindent\textbf{Hybrid 7} This is the simulator's distribution.

Hybrid 6 and 7 are indistinguishable because the inputs are distributed exactly the same (they are all $0$'s).

This completes our proof.




\end{proof}

\fi

\end{document}